\documentclass[11pt]{article}
\usepackage{dsfont}
\usepackage{amsmath}
\usepackage{amsthm}
\usepackage{amssymb}
\usepackage{algorithm,algorithmic}
\usepackage{graphicx}
\usepackage{url}
\usepackage{hyperref}
\usepackage[letterpaper,margin=1in]{geometry}
\usepackage{authblk}
\usepackage{tikz} 
\usepackage{subcaption}
\usepackage{slashbox}

\usepackage{booktabs}
\usepackage{multirow}

\usetikzlibrary{arrows,decorations.markings}

\usetikzlibrary{shapes.misc}
\usetikzlibrary{calc}

\tikzset{cross/.style={cross out, draw=black, minimum size=2*(#1-\pgflinewidth), inner sep=0pt, outer sep=0pt},
cross/.default={1pt}}

\usepackage{color}

\definecolor{blueblack}{rgb}{0,0,.7}

\newcounter{sideremark}

\theoremstyle{plain}
\newtheorem{theorem}{Theorem}[section]
\newtheorem{lemma}[theorem]{Lemma}

\newtheorem{remark}[theorem]{Remark}

\newtheorem{corollary}[theorem]{Corollary}

\theoremstyle{remark}

\newtheorem*{note}{Note}

\def\eg{{\it e.g.,}~}
\def\ie{{\it i.e.,}~}

\makeatletter
\newtheorem*{rep@theorem}{\rep@title}
\newcommand{\newreptheorem}[2]{%
\newenvironment{rep#1}[1]{%
 \def\rep@title{\emph{\textbf{#2} \ref{##1}}}%
 \begin{rep@theorem}}%
 {\end{rep@theorem}}}
\makeatother
\newreptheorem{theorem}{Theorem}
\newreptheorem{lemma}{Lemma}
\newreptheorem{proposition}{Proposition}
\usepackage{thmtools,thm-restate}

\newif\ifshort
\shorttrue 

\newcommand{\Prob}[1]{{\mathbf{Pr}\left[ #1 \right]}}
\newcommand{\Ex}[1]{\ensuremath{\mathbf{E}\left[#1\right]}}
\newcommand{\Var}[1]{\ensuremath{\mathbf{Var}\left[#1\right]}}

\DeclareMathOperator{\SIM}{sim}

\DeclareMathOperator{\lsb}{lsb}

\newcommand{\bset}{\lbrace 0,1 \rbrace}

\begin{document}
%
\title{Similarity Search for Dynamic Data Streams}
\date{}
%
%
%
%


\author{Marc~Bury, 
Chris Schwiegelshohn,~Mara~Sorella
\thanks{This work was supported by ERC Advanced Grant 788893 AMDROMA.}
}

%
%

\markboth{Journal of Transactions on Knowledge and Data Engineering}{}
%




\maketitle

\begin{abstract}
Nearest neighbor searching systems are an integral part of many online applications, including but not limited to pattern recognition, plagiarism detection and recommender systems. 
With increasingly larger data sets, scalability has become an important issue. Many of the most space and running time efficient algorithms are based on locality sensitive hashing.
Here, we view the data set as an $n$ by $|U|$ matrix where each row corresponds to one of $n$ users and the columns correspond to items drawn from a universe $U$.
The de facto standard approach to quickly answer nearest neighbor queries on such a data set is usually a form of min-hashing.
Not only is min-hashing very fast, but it is also space efficient and can be implemented in many computational models aimed at dealing with large data sets such as MapReduce and streaming.
However, a significant drawback is that minhashing and related methods are only able to handle insertions to user profiles and tend to perform poorly when items may be removed.
We initiate the study of scalable locality sensitive hashing (LSH) for fully dynamic data-streams. Specifically, using the Jaccard index as similarity measure, we design (1) a collaborative filtering mechanism maintainable in dynamic data streams and (2) a sketching algorithm for similarity estimation.
Our algorithms have little overhead in terms of running time compared to previous LSH approaches for the insertion only case, and drastically outperform previous algorithms in case of deletions.
\end{abstract}




%

\section{Introduction}\label{sec:introduction}

%
%
%
%
Finding the most interesting pairs of points, i.e. typically those having small distance or, conversely, of high similarity, known as the \emph{nearest-neighbor search} problem, is a task of primary importance, that has many applications such as plagiarism detection~\cite{BGMZ97}, clustering~\cite{GRS00}, association rule mining~\cite{CDFGIMUY01}.
The aim is to maintain a data structure such that we can efficiently report all neighbors within a certain distance from a candidate point.
\emph{Collaborative filtering}~\cite{SLH14} is an approach to produce such an item set by basing the recommendation on the most similar users in the data set and suggesting items not contained in the intersection.
To apply such an approach, one typically requires two things: (1) a measure of similarity (or dissimilarity) between users and (2) scalable algorithms for evaluating these similarities.
In these contexts scalability can mean fast running times, but can also require strict space constraints.

Though it is by no means the only method employed in this line of research, locality senstive hashing (LSH) satisfies both requirements~\cite{CDFGIMUY01,DDGR07}. For a given similarity measure, the algorithm maintains a small number of hash-values, or fingerprints that represent user behavior in a succinct way. The name implies that the fingerprints have locality properties, i.e., similar users have a higher probability of sharing the same fingerprint whereas dissimilar users have a small probability of agreeing on a fingerprint.
The fingerprints themselves allow the recommendation system to quickly filter out user pairs with low similarity, leading to running times that are almost linear in input and output size.

A crucial property of LSH-families is that they are data-oblivious, that is the properties of the hash family depend only on the similarity measure but not on the data. Therefore, LSH-based filtering can be easily facilitated in online and streaming models of computation, where user attributes are added one by one in an arbitrary order. The fingerprint computation may fail, however, if certain attributes get deleted.
Attribute deletion occurs frequently, for instance, if the data set evolves over time. Amazon allows users to unmark certain bought items for recommendations, Twitter users have an unfollow option, Last.fm users may delete songs or artists from the library. A naive way to incorporate deletions within the context of LSH is to recompute any affected fingerprint, which requires scanning the entire user profile and is clearly infeasible.
\subsection{Contributions}
We initiate the study of locality sensitive nearest neighbors search in the dynamic data-stream model. Our input consists of sequence of triples $(i,j,k)$, where $i\in [n]$ is the user identifier, $j\in [|U|]$ is the item identifier and $k\in\{-1,1\}$ signifying insertion or deletion. Instead of maintaining an $n\times |U|$ user/attribute matrix, we keep a sketch of $\text{polylog}(n \cdot |U|)$ bits per user. 

In a first step, we show that the Jaccard distance $1-\frac{|A\cap B|}{|A\cup B|}$ can be $(1 \pm \varepsilon)$-approximated in dynamic streams. 
Moreover, the compression used in this approximation is a black-box application of $\ell_0$ sketches, which allows for extremely efficient algorithms in both theory and practice.
This also enables us to efficiently compress the $n$ by $n$ distance matrix using $n\text{ polylog }(n|U|)$ bits, similar in spirit to the compression by Indyk and Wagner~\cite{IndykW17} for Euclidean spaces. This enables us to run any distance matrix-based algorithm in a dynamic semi-streaming setting.

Known lower bounds on space complexity of set intersection prevent us from achieving a compression with multiplicative approximation ratio for Jaccard similarity, see for instance \cite{PSW14}.
From the multiplicative approximation for Jaccard distance we nevertheless get an $\varepsilon$-additive approximation to Jaccard similarity, which may be sufficient if the interesting similarities are assumed to exceed a given threshold.
However, even with this assumption, such a compression falls short of the efficiency we are aiming for, as it is not clear whether the relevant similarities can be found more quickly than by evaluating all similarities.

Our main contribution lies now in developing a compression scheme that simultaneously supports locality-sensitive hashing while satisfying a weaker form of approximation ratio.
The construction is inspired by bit-hashing techniques used both by $\ell_0$ sketches and min-hashing.
In addition, our approach can be extended to other similarities admitting LSHs other than min-hashing, such as Hamming, Anderberg, and Rogers-Tanimoto similarities.
This approach, despite having provable bounds that are weaker than $\ell_0$ sketches from an approximation point of view, is extremely simple to implement. Our implementation further showed that our algorithms have none to little overhead in terms of running time compared to previous LSH approaches for the insertion only case, and drastically outperform previous algorithms in case of deletions.


\section{Preliminaries}
\label{sec:prelim}

We have $n$ users. A user profile is a subset of some universe $U$. 
The symmetric difference of two sets $A,B\subseteq U$ is $A \bigtriangleup B = (A \setminus B) \cup (B \setminus A)$. The complement is denoted by $\overline{A} = U\setminus A$.
Given $x,y\geq 0$ and $0\leq z\leq z'$, the \emph{rational set similarity} $S_{x,y,z,z'}$ between two item sets $A$ and $B$ is
\[S_{x,y,z,z'}(A,B)=\frac{x\cdot \vert A\cap B \vert + y \cdot \vert \overline{A\cup B}\vert + z\cdot \vert A\bigtriangleup B\vert}{x\cdot \vert A\cap B \vert + y \cdot \vert \overline{A\cup B}\vert + z'\cdot \vert A\bigtriangleup B\vert}\]
if it is defined and $1$ otherwise.
The distance function induced by a similarity $S_{x,y,z,z'}$ is defined as $D_{x,y,z,z'}(A,B):=1-S_{x,y,z,z'}(A,B)$.
If $D_{x,y,z,z'}$ is a metric, we call $S_{x,y,z,z'}$ a metric rational set similarity~\cite{Jan06}. 
The arguably most well-known rational set similarity is the Jaccard index $S(A,B)=S_{1,0,0,1}(A,B)=\frac{|A\cap B|}{|A\cup B|}$.
A \emph{root similarity} is defined as $S_{x,y,z,z'}^{\alpha} := 1-(1-S_{x,y,z,z'})^{\alpha}$ for any $0<\alpha\leq 1$.
We denote numerator and denominator of a rational set similarity by $Num(A,B)$ and $Den(A,B)$, respectively.
For some arbitrary but fixed order of the elements, we represent $A$ via its characteristic vector $a \in \{0,1\}^{|U|}$ with $a_i = 1$ iff $i \in A$. The $\ell_p$-norm of a vector $a \in \mathbb{R}^d$ is defined as $\ell_p(a)=\sqrt[p]{\sum_{i=1}^d |a_i|^p}$. Taking the limit of $p$ to $0$, $\ell_0(x)$ is exactly the number of non-zero entries, i.e. $\ell_0(a) = \vert \lbrace i \mid a_i \neq 0 \rbrace \vert$.

An LSH for a similarity measure $S:U\times U\rightarrow [0,1]$ is a set of hash functions $H$ on $U$ with an associated probability distribution such that $\Prob{h(A)=h(B)} = S(A,B)$
for $h$ drawn from $H$ and any two item sets $A,B\subseteq U$.
We will state our results in a slightly different manner.
A $(r_1,r_2,p_1,p_2)$-sensitive hashing scheme for a similarity measure aims to find a distribution over a family of hash functions $H$ such that for $h$ drawn from $H$ and two item sets $A,B\subseteq U$ we have $\Prob{h(A)=h(B)} \geq p_1$ if $ S(A,B)\geq r_1$ and $\Prob{h(A)=h(B)} \leq p_2$ if $ S(A,B)\leq r_2$. 
The former definition due to Charikar~\cite{Cha02} has a number of appealing properties and is a special case of the latter definition due to Indyk and Motwani~\cite{InM98}.
Unfortunately, it is also a very strong condition and in fact not achievable for dynamic data streams.
We emphasize that the general notions behind both definitions are essentially the same.

\section{Related Work}
\label{sec:related}
\subsection*{Locality Sensitive Hashing} 
Locality sensitive hashing describes an algorithmic framework for fast (approximate) nearest neighbor search in metric spaces.
In the seminal paper by Indyk and Motwani~\cite{InM98}, it was proposed as a way of coping with the \emph{curse of dimensionality} for proximity problems in high-dimensional Euclidean spaces.
The later, simpler definition by Charikar~\cite{Cha02} was used even earlier in the context of min-hashing for the Jaccard index by Broder et al.~\cite{Bro97,Bro00,BCFM00}. 
Roughly speaking, min-hashing computes a fingerprint of a binary vector by permuting the entries and storing the first non-zero entry. 
For two item sets $A$ and $B$, the probability that the fingerprint is identical is equal to the Jaccard similarity of $A$ and $B$. 
When looking for item sets similar to some set $A$, one can arrange multiple fingerprints to filter out sets of small similarity while retaining sets of high similarity, see Cohen et al.~\cite{CDFGIMUY01}, and Leskovec et al.~\cite{LRU14} for details.
We note that while this paper is focused mainly on min-hashing, locality sensitive hashing has been applied to many different metrics, see Andoni and Indyk~\cite{AnI08} for an overview.

Instead of using multiple independent hash functions to generate $k$ fingerprints, Cohen and Kaplan suggested using the $k$ smallest entries after a single evaluation~\cite{CoK07,CoK07b} which is known as bottom $k$-sampling, see also Hellerstein et al.~\cite{HHW97}.
Min-hashing itself is still an active area of research. Broder et al.~\cite{BCFM00} showed that an ideal min-hash family is infeasible to store, which initiated the search for more feasible alternatives. Indyk considered families of approximate min-wise indepedence~\cite{Ind01}, i.e. the probability of an item becoming the minimum is not uniform, but close to uniform, see also Feigenblat et al~\cite{FPS11}.

Instead of basing requirements on the hash-function, other papers focus on what guarantees are achievable by simpler, easily stored and evaluated hash-functions with limited independence.
Of particular interest are 2-wise independent hash functions.
Dietzfelbinger~\cite{Die96} showed that, given two random numbers $a$ and $b$, the hash of $x$ may be computed via $(ax+b)\gg k$, where $k$ is a power of $2$ and $\gg$ denotes a bit shift operation.
This construction is to the best of our knowledge the fastest available and there exists theoretical evidence which supports that it may be, in many cases, good enough.
Chung et al.~\cite{CMV13} showed that if the entropy of the input is large enough, the bias incurred by $2$-wise independent hash functions becomes negligible.
Thorup~\cite{Tho13} further showed that $2$-wise independent hashing may be used for a bottom $k$ sampling with a relative error of $\frac{1}{\sqrt{fk}}$, where $f$ is the Jaccard similarity between two items.
Better bounds and/or running times are possible using more involved hash functions such as tabulation hashing~\cite{PaT13,Tho13b}, linear probing~\cite{PPR11,PaT16}, one-permuatation hashing~\cite{DKT17a,LOZ12,SL14}, and feature hashing~\cite{DKT17b,WDLSA09}. 

\subsection*{Profile Sketching}
Using the index of an item as a fingerprint is immediate and requires $\log |U|$ space.
For two item sets $A,B$, we then require roughly $\frac{\log |U|}{\varepsilon^2\cdot S(A,B)}$ bits of space to get an $(1\pm \varepsilon)$-approximate estimate of the similarity $S(A,B)$.
It turns out that this is not optimal,
Bachrach and Porat~\cite{BP15,BaP13} and Li and K\"onig~\cite{LiK11} proposed several improvements and constant size fingerprints are now known to exist.
Complementing these upper bounds are lower bounds by Pagh~\cite{PSW14} who showed that that this is essentially optimal for summarizing Jaccard similarity.

We note that one of the algorithms proposed by Bachrach and Porat in~\cite{BP15} use an $\ell_2$ estimation algorithm as a black box to achieve fingerprint size of size $\frac{(1-S(A,B))^2}{\varepsilon^2}\log |U|$ bits. It is well known that the $\ell_2$ norm of a vector can be maintained in dynamic data streams. However, their algorithm only seems to work for if the similarity is sufficiently large, i.e. $S(A,B)\geq 0.5$ and it does not seem to support locality sensitive hashing.

\section{Similarity sketching in dynamic Streams}
\label{sec:est}
In this section, we aim to show that the distance function of any rational set similarity with an LSH can be $(1\pm \epsilon)$-approximated in dynamic streams. First, we recall the following theorem relating LSH-ability and properties of the induced dissimilarity:
\begin{theorem}
\label{thm:characterization}
Let $x,y,z,z'>0$. Then the following three statements are equivalent.
\begin{enumerate}
\item $S_{x,y,z,z'}$ has an LSH.
\item $1-S_{x,y,z,z'}$ is a metric.
\item $z'\geq\max (x,y,z)$.
\end{enumerate}
\end{theorem}

(1)$\Rightarrow$(2) was shown by Charikar~\cite{Cha02}, (2)$\Rightarrow$(1) was shown by Chierichetti and Kumar~\cite{ChK15} and (2)$\Leftrightarrow$(3) was proven by Janssens~\cite{Jan06}. 
We also recall the state of the art of $\ell_0$ sketching in dynamic streams.
\begin{theorem}[Th. 10 of Kane, Nelson, and Woodruff~\cite{KNW10}]
\label{thm:knw}
There is a dynamic streaming algorithm for $(1\pm \varepsilon)$-approximating $\ell_0(x)$ of a $|U|$-dimensional vector $x$ using space $O(\frac{1}{\varepsilon^2}\log |U| )$\footnote{\footnotesize The exact space bounds of the $\ell_0$ sketch by Kane, Nelson and Woodruff depends on the magnitude of the entries of the vector. The stated space bound is sufficient for our purposes as we are processing binary entries.}, with probability $2/3$, and with $O(1)$ update and query time.
\end{theorem}

With this characterization, we prove the following.

\begin{theorem}
\label{thm:streamdist}
Given a constant $0<\varepsilon\leq 0.5$, two item sets $A,B\subseteq U$ and some rational set similarity $S_{x,y,z,z'}$ with metric distance function $1-S_{x,y,z,z'}$, there exists a dynamic streaming algorithm that maintains a $(1\pm \varepsilon)$ approximation to $1-S_{x,y,z,z'}(A,B)$ with constant probability. The algorithm uses $O(\frac{1}{\varepsilon^2}\log \lvert U \rvert)$ space and has $O(1)$ update and query time.
\end{theorem}
\begin{proof}
We start with the observation that $\vert A\bigtriangleup B\vert = \ell_0(a-b)$ and $\vert A\cup B\vert = \ell_0(a+b)$, where $a$ and $b$ are the characteristic vectors of $A$ and $B$, respectively. Since $Den(A,B) - Num(A,B) = (z'-z) \cdot |A\bigtriangleup B|$ is always non-negative due to $z'\geq z$, we only have to prove that $Den(A,B)$ is always a non-negative linear combination of terms that we can approximate via sketches.
First, consider the case $x\geq y$.
Reformulating $Den(A,B)$, we have
\[Den(A,B)  = y\cdot |U| + (x-y)\cdot \vert A\cup B\vert + (z'-x)\cdot \vert A\bigtriangleup B|.\]
Then both numerator and denominator of $1-S_{x,y,z,z'}$ can be written as a non-negative linear combination of $n$, $\vert A\bigtriangleup B\vert$ and $\vert A\cup B\vert$. Given a $(1\pm \varepsilon)$ of these terms, we have an upper bound of $\frac{1+\varepsilon}{1-\varepsilon} \leq (1+\varepsilon)\cdot(1+2\varepsilon) \leq (1+5\varepsilon)$ and a lower bound of $\frac{1-\varepsilon}{1+\varepsilon}\geq (1-\varepsilon)^2 \geq (1-2\varepsilon)$ for any $\varepsilon\leq 0.5$. 

Now consider the case $x<y$. We first observe
$S_{x,y,z,z'}(A,B)=S_{y,x,z,z'}(\overline{A},\overline{B}).$
Therefore
\begin{equation}
\nonumber
Den(A,B) =  (y-x)\cdot \vert \overline{A}\cup \overline{B}\vert + x\cdot |U|  + (z'-y)\cdot \vert \overline{A} \bigtriangleup \overline{B} \vert.
\end{equation}
Again, we can write the denominator as a non-negative linear combination of $\vert \overline{A} \bigtriangleup \overline{B} \vert$, $n$ and $\vert \overline{A} \cup \overline{B}\vert$. Dynamic updates can maintain an approximation of $\vert \overline{A} \bigtriangleup \overline{B} \vert$ and $\vert \overline{A} \cup \overline{B}\vert$, leading to upper and lower bounds on the approximation ratio analogous to those from case $x\geq y$.

By plugging in the $\ell_0$ sketch of Theorem~\ref{thm:knw} and rescaling $\varepsilon$ by a factor of $5$, the theorem follows. 
\end{proof}

Using a similar approach, we can approximate the distance of root similarity functions admitting a locality hashing scheme.
We first repeat the following characterization.

\begin{theorem}[Th. 4.8 and 4.9 of~\cite{ChK15}]
\label{thm:rootchar}
The root similarity $S_{x,y,z,z'}^{\alpha}$ is LSHable if and only if $z'\geq \frac{\alpha+1}{2}\max(x,y)$ and $z'\geq z$.
\end{theorem}

\begin{theorem}
\label{thm:rootdist}
Given a constant $0<\varepsilon\leq 0.5$, two item sets $A,B\subseteq U$ and some LSHable root similarity $S_{x,y,z,z'}^{\alpha}$, there exists a dynamic streaming algorithm that maintains a $(1\pm \varepsilon)$ approximation to $1-S_{x,y,z,z'}^{\alpha}(A,B)$ with constant probability. The algorithm uses $O(\frac{1}{\varepsilon^2}\log \lvert U \rvert)$ space and each update and query requires $O(1)$ time.
\end{theorem}
\begin{proof}
We consider the case $x\geq y$, the case $y\geq x$ can be treated analogously.
Again we will show that we can $(1\pm\varepsilon)$-approximate the denominator; the remaining arguments are identical to those of Theorem~\ref{thm:streamdist}.
Consider the following reformulation of the denominator 
\[Den(A,B)  = y\cdot n + (x-z')\cdot \vert A\cap B\vert + (z'-y)\cdot \vert A \cup B \vert.\]

We first note that we can obtain an estimate of $|A\cap B|$ in a dynamic data stream with additive approximation factor $\varepsilon\cdot |A\cup B|$ by computing $|A| + |B| - \widehat{|A\cup B|}$, where $\widehat{|A\cup B|}$ is a $(1\pm \varepsilon)$-approximation of $|A\cup B|$.

Due to Theorem~\ref{thm:rootchar}, we have $x-z'\leq 2\cdot z' -z' \leq z'$ and either $z'-y \geq  \frac{z'}{2}$ or $y\geq \frac{z'}{2}$. Hence $\varepsilon\cdot (x-z') \leq \varepsilon\cdot z' \leq 2\varepsilon \cdot \max(z',(z'-y))$.
Since further $|U|\geq |A\cup B|$, we then obtain a $(1\pm 2\varepsilon)$-approximation to the denominator. Rescaling $\varepsilon$ completes the proof.
\end{proof}

\begin{remark}
Theorems~\ref{thm:streamdist} and~\ref{thm:rootdist} are not a complete characterization of divergences  induced by similarities that can be $(1\pm\varepsilon)$-approximated in dynamic streams.  
Consider, for instance, the S{\o}renson-Dice coefficient $S_{2,0,0,1} = \frac{2\cdot |A\cap B|}{|A|+|B|}$ with $1-S_{2,0,0,1} = \frac{|A\bigtriangleup B|}{|A|+|B|}$.
Neither is $1-S_{2,0,0,1}$ a metric, nor do we have $z' \geq \frac{\alpha+1}{2} x$ for any $\alpha>0$. 
However, both numerator and denominator can be approximated using $\ell_0$ sketches.
\end{remark}

The probability of success can be further amplified to $1-\delta$ in the standard way by taking the median estimate of $O(\log(1/\delta))$ independent repetitions of the algorithm. For $n$ item sets, and setting $\delta=1/n^2$, we then get the following.
\begin{corollary}
\label{cor:streamapprox}
Let $S$ be a rational set similarity with metric distance function $1-S$.
Given a dynamic data stream consisting of updates of the form $(i,j,v) \in [n] \times [|U|] \times \lbrace -1, +1 \rbrace$ meaning that $a^{(i)}_j = a^{(i)}_j+v$ where $a^{(i)} \in \bset^{|U|}$ with $i\in\{ 1,\ldots,n\}$, there is a streaming algorithm that can compute with constant probability for all pairs $(i,i')$

\noindent$\bullet$ a $(1\pm \varepsilon)$ multiplicative approximation of $1-S(a^i,a^{i'})$  and

\noindent$\bullet$ an $\epsilon$-additive approximation of $S(a^i,a^{i'})$.
\\
\noindent The algorithm uses $O(n\log n \cdot \varepsilon^{-2} \cdot \log \lvert U \rvert)$ space and each update and query needs $O(\log n)$ time.
\end{corollary}

We note that despite the characterization of LSHable rational set similarities of Theorem~\ref{thm:characterization}, the existence of the  approximations of Corollary~\ref{cor:streamapprox} hints at, but does not directly imply the existence of a locality sensitive hashing scheme or even an approximate locality sensitive hashing scheme on the sketched data matrix in dynamic streams.
Our second and main contribution now lies in the design of a simple LSH scheme maintainable in dynamic data streams, albeit with weaker approximation ratios. The scheme is space efficient, easy to implement and to the best of our knowledge the first of its kind able to process deletions.

\begin{remark}
Corollary~\ref{cor:streamapprox} also implies that any algorithm based on the pairwise distances of a rational set similarity admits a dynamic streaming algorithm using $n\cdot \text{ polylog }(n|U|)$ bits of space. Notable examples include hierarchical clustering algorithms such as single or complete linkage, distance matrix methods used in phylogeny, and visualization methods such as heatmaps. Though the main focus in the experimental section (Section~\ref{sec:experiments}) will be an evaluation of the dynamic hashing performance, we also briefly explore clustering and visualization methods based on the sketched distance matrix.
\end{remark}

\section{An LSH algorithm for dynamic data streams}

\label{sec:filter}

In the following, we will present a simple dynamic streaming algorithm that supports Indyk and Motwani-type sensitivity.
Recall that we want to find pairs of users with similarity greater than a parameter $r_1$, while we do not want to report pairs with similarity less than $r_2$. 
The precise statement is given via the following theorem.

\begin{theorem}
\label{thm:minhash_dynamic}
Let $0 < \varepsilon, \delta, r_1,r_2 < 1$ be parameters.
Given a dynamic data stream with $n$ users and $|U|$ attributes, there exists an algorithm that maintains a $(r_1,r_2,(1-\varepsilon)r_1,6r_2 /(\delta(1-\varepsilon/5 \sqrt{2r_1}))$-sensitive LSH for Jaccard similarity with probability $1-\delta$.
For each user,  $O(\frac{1}{\varepsilon^4\delta^5\cdot r_1^2}\log^2 |U|)$ bits of space are sufficient. The update time is $O(1)$.
\end{theorem}

The proof of this theorem consists of two parts. First, we give a probabilistic lemma from which we derive the sensitivity parameters. Second, we describe how the sampling procedure can be implemented in a streaming setting.

\subsection{Sensitivity Bounds}
While a black box reduction from any $\ell_0$ sketch seems unlikely, we note that most $\ell_0$ algorithms are based on bit-sampling techniques similar to those found in min-hashing. 
Our own algorithm is similarly based on sampling a sufficient number of bits or item indexes from each item set. Given a suitably filtered set of candidates, these indexes are then sufficient to infer the similarity. 
Let $U_k\subseteq U$ be a random set of elements where each element is included with probability $2^{-k}$. Further, for any item set $A$, let $A_k=A\cap U_k$. Note that in $S_{x,y,z,z'}(A_k,B_k)$ the value of $|U|$ is replaced by $|U_k|$. 
At the heart of the algorithm now lies the following technical lemma.
\begin{lemma}
\label{lem:samplesim_general}
Let $0<\varepsilon,\delta,r <1$ be constants and $S_{x,y,z,z'}$ be a rational set similarity with metric distance function. 
Let $A$ and $B$ be two item sets.
Assume we sample every item uniformly at random with probability $2^{-k}$, where $k  \leq \log \left( \dfrac{\varepsilon^2\cdot \delta\cdot r\cdot Den_{x,y,z,z'}(A,B)}{100\cdot z'}\right )$.
Then with probability at least $1-\delta$ the following two statements hold.
\begin{enumerate}
\item If $S_{x,y,z,z'}(A,B) \geq r$ we have
$ (1-\varepsilon)S_{x,y,z,z'}(A,B) \leq S(A_k, B_k) \leq (1+\varepsilon)S_{x,y,z,z'}(A,B).$
\item $S_{x,y,z,z'}(A_k, B_k)\leq \dfrac{2\cdot S_{x,y,z,z'}(A,B)}{\delta(1-(\varepsilon/5) \cdot \sqrt{2r})}.$
\end{enumerate}
\end{lemma}

We note that any metric distance function induced by a rational set similarity satisfies $z'\geq \max(x,y,z)$, see Theorem~\ref{thm:characterization} in Section~\ref{sec:est}.
\begin{proof}
Let $Den_k = Den(A_k,B_k)$, $Num_k = Num(A_k,B_k)$, and $X_i = 1$ iff $i \in U_k$.
If $S_{x,y,z,z}(A,B) \geq r$ then $Num(A,B) \geq r \cdot Den(A,B)$. 
Thus, we have $\mathbb{E}[Num_k] = Num(A,B) / 2^{k} \geq r\cdot Den(A,B) / 2^{k}$ and $\mathbb{E}[Den_k] = Den(A,B) / 2^{k}$. 
Further, we have $\Var{X_i} = 2^{-k}\cdot(1-2^{-k})\leq 2^{-k}$. for any $X_i$.
We first give a variance bound on the denominator.
\begin{eqnarray*}
& & \Var{Den_k}\\
& = & \mathbf{Var}[x \cdot \vert A_k \cap B_k \vert + y\cdot(\vert U_k \vert - |A_k\cup B_k|) | \\
& &+ z' \cdot \vert A_k \bigtriangleup B_k \vert]\\
&= & x^2\sum_{i \in A \cap B} \Var{X_i} + 
 y^2 \sum _{i \in \overline{A \cup B}}\Var{X_i} \\
 & &+ z'^2\sum_{i \in A \bigtriangleup B}\Var{X_i} \\
&= & \left( (x^2-y^2) \vert A \cap B \vert + y^2 \cdot |U| \right.\\
& & \left. + (z'^2-y^2) \vert A \bigtriangleup B \vert \right)\Var{X_i} \\
&\leq & \left( (x^2-y^2) \vert A \cap B \vert + y^2 \cdot |U| + (z'^2-y^2) \vert A \bigtriangleup B \vert \right)/2^k \\
&\leq &\frac{1}{2^k}\max\lbrace x+y, z'+y, y \rbrace\cdot ( (x-y) \vert A \cap B \vert + \\
& & y \cdot d + (z'-y) \vert A \bigtriangleup B \vert) \leq  2z' \cdot \Ex{Den_k}
\end{eqnarray*}
and analogously
$$\Var{Num_k} \leq \max\lbrace x+y, z+y, y \rbrace \cdot \Ex{Num_k}.$$
Using Chebyshev's inequality we have
\begin{eqnarray*}
& &\mathbb{P}\left[ Den_k - \Ex{Den_k}  \geq \varepsilon/5\cdot \Ex{Den_k}\right] \\
& \leq &  \dfrac{50z'}{\varepsilon^2\cdot \Ex{Den_k}} \leq  \dfrac{50z'\cdot 2^k}{\varepsilon^2 \cdot Num(A,B)}, 
\end{eqnarray*}
and
\begin{eqnarray*}
& &\mathbb{P}\left[ Num_k - \Ex{Num_k} \geq \varepsilon/5\cdot \Ex{Num_k}\right] \\
& \leq &  \dfrac{25\max\lbrace x+y, z+y, y \rbrace}{\varepsilon^2\cdot \Ex{Num_k}}  \leq  \dfrac{50z' \cdot 2^k}{\varepsilon^2 \cdot Num(A,B)}. 
\end{eqnarray*}
If $k \leq \log \left( \dfrac{\varepsilon^2 \delta r Den(A,B)}{100z'}\right)\leq \log \left( \dfrac{\varepsilon^2 \delta Num(A,B)}{100z'} \right)$ then both $ Den_k - \Ex{Den_k} \leq \frac{\varepsilon}{5} \Ex{Den_k}$ and $ Num_k - \Ex{Num_k}  \leq \frac{\varepsilon}{5} \Ex{Num_k}$ hold with probability at least $1-\delta/2$. 
Then we can bound $S(A_k,B_k) = Num_k/Den_k$ from above by 
\begin{eqnarray*}
& &\dfrac{Num(A,B) / 2^k + \varepsilon Num(A,B)}{Den(A,B)/2^k-\varepsilon Den(A,B) / 2^k)} \\
&=&\dfrac{1+\varepsilon/5}{1-\varepsilon/5} \cdot S_{x,y,z,z'}(A,B) \leq (1+\varepsilon)\cdot S_{x,y,z,z'}(A,B).  
\end{eqnarray*}
Analogously, we can bound $S_{x,y,z,z'}(A_k,B_k)$ from below by $\frac{1-\varepsilon/5}{1+\varepsilon/5} \cdot S_{x,y,z,z'}(A,B)\geq (1-\varepsilon)\cdot S_{x,y,z,z'}(A,B)$ which concludes the proof of the first statement. 

For the second statement, we note that the expectation of $Num_k$ can be very small because we have no lower bound on the similarity.  Hence, we cannot use Chebyshev's inequality for an upper bound on  $Num_k$. 
But it is enough to bound the probability that $Num_k$ is greater than or equal to $(2/\delta) \cdot \Ex{Num_k}$ by $\delta/2$ using Markov's inequality. 
With the same arguments as above, we have that the probability of $Den_k \leq (1-\varepsilon') \cdot \Ex{Den_k}$ is bounded by $\frac{\varepsilon^2 r \delta}{25\cdot\varepsilon'^2}$ which is equal to $\delta/2$ if $\varepsilon' = \varepsilon/5 \cdot \sqrt{2r}$.  
Putting everything together we have that 
$$ S_{x,y,z,z'}(A_k,B_k) \leq \dfrac{2}{\delta(1-(\varepsilon/5) \cdot \sqrt{2r})} \cdot S_{x,y,z,z'}(A,B) $$
with probability at least $1-\delta$.
\end{proof}

We note that for similarities with $y>x$, we can obtain the same bounds by sampling $0$-entries instead of $1$-entries. Since we are not aware of any similarities with this property used in practice, we limited our analysis to the arguably more intuitive case $x\geq y$.

Applying this lemma on a few better known similarities gives us the following corollary. 
We note that to detect candidate high similarity pairs for an item set $A$, $Den:=|A\cup B|\geq |A|$ for Jaccard  and $Den:=|A\cup B| + |A\bigtriangleup B| \geq |A|$ for Anderberg. For Hamming and Rogers-Tanimoto similarities, $Den\geq |U|$.
More examples of rational set similarities can be found in Naish, Lee, and Ramamohanarao~\cite{NLR11}.
\begin{corollary}
\label{cor:simsample}
For the following similarities, the following values of $k$ are sufficient to apply Lemma \ref{lem:samplesim_general}:
\begin{center}
\begin{tabular}{c||c|c}
Similarity & Parameters & Sampling Rate \\
\hline\hline
Jaccard & $S_{1,0,0,1}$ & $\log \left( \varepsilon^2 \delta r \vert A\vert/100  \right)$\\
\hline
Hamming & $S_{1,1,0,1}$ & $\log \left( \varepsilon^2 \delta r \vert U\vert/100  \right)$\\
\hline
Anderberg & $S_{1,0,0,2}$ & $\log \left( \varepsilon^2 \delta r \vert A\vert/200  \right)$ \\
\hline
Rogers-Tanimoto & $S_{1,1,0,2}$& $\log \left( \varepsilon^2 \delta r \vert U\vert/200  \right)$
\end{tabular}
\end{center}
\end{corollary}

\subsection{Streaming Implementation}
When applying Lemma~\ref{lem:samplesim_general} to a dynamic streaming environment, we have to address a few problems. 
First, we may not know how to specify the number of items we are required to sample.
For Hamming and Rogers-Tanimoto similarities, it is already possible to run a black box LSH algorithm (such as the one by Cohen et al.~\cite{CDFGIMUY01}) if the number of sampled items are chosen via Corollary~\ref{cor:simsample}.
For Jaccard (and Anderberg), the sample sizes depend on the cardinality of $A$, which requires additional preprocessing steps.

\subsection*{Cardinality-Based Filtering}
As a first filter, we limit the candidate solutions based on their respective supports.
For each item, we maintain the cardinality, which can be done exactly in a dynamic stream via counting. If the sizes of two item sets $A$ and $B$ differ by a factor of at least $r_1$, \ie $\vert A \vert \geq r_1 \cdot \vert B \vert$, then the distance between these two sets has to be 
$$ 1-S(A,B) = \dfrac{\vert A \bigtriangleup B \vert}{\vert A \cup B \vert} \geq \dfrac{\vert A \vert - \vert B \vert}{\vert A \vert} \geq 1-1/r_1. $$
We then discard any item set with cardinality not in the range of $[r_1\cdot |A|, |A|]$.
Like the algorithm by Cohen et al~\cite{CDFGIMUY01}, we can do this by sorting the rows or hashing. 

\subsection*{Small Space Item Sampling}
Since the cardinality of an item set may increase and decrease as the stream is processed, we have to maintain multiple samples $U_k$ in parallel for various values of $k$. If a candidate $k$ is larger than the threshold given by Corollary~\ref{cor:simsample}, we will sample only few items and still meet a small space requirement.
If $k$ is too small, $|U_k|$ might be too large to store.
We circumvent this using a nested hashing approach we now describe in detail.

We first note that $U_k$ does not have to be a fully independent randomly chosen set of items. 
Instead, we only require that the events $X_i$ are pairwise independent.
The only parts of the analysis of Lemma~\ref{lem:samplesim_general} that could be affected are the bounds on the variances, which continue to hold for pairwise independence.
This allows us to emulate the sampling procedure using universal hashing.
Assume that $M$ is a power of $2$ and let $h:[|U|]\rightarrow [M]$ be a $2$-wise independent universal hash function, i.e. $\mathbb{P}[h(a)=j] = \frac{1}{M}$, for all $j\in [M]$.
We set $U_k = \lbrace j \in [M]~|~ \lsb(h(j)) = k \rbrace$, where $\lsb(x)$ denotes the first non-zero index of $x$ when $x$ is written in binary and $\lsb(0) = \log M$.
Since the image of $h$ is uniformly distributed on $[M]$, each bit of $h(j)$ is $1$ with probability $1/2$, and hence we have $\mathbb{P}[\lsb(h(j))=k] = 2^{-k}$.
Moreover, for any two $j,j'$ the events that $\lsb(h(j))=k$ and $\lsb(h(j'))=k$ are independent.
The value of $M$ may be adjusted for finer ($M$ large) or coarser ($M$ small) sampling probabilities. In our implementation (see Section~\ref{sec:experiments}) as well as in the proof of Theorem~\ref{thm:minhash_dynamic}, we set $M=|U|$.
 Following Dietzfelbinger~\cite{Die96}, $h$ requires $\log |U|$ bits of space.

To avoid storing the entire domain of $h$ in the case of large $|U_k|$, we pick, for each $k\in[0, \dots, \log |U|]$, another $2$-wise independent universal hash function $h_k:[|U|]\rightarrow [c^2]$, for some absolute constant $c$ to be specified later.
For some $j\in [|U|]$, we first check if $\lsb(h(j))=k$.
If this is true, we apply $h_k(j)$.

For the $i$th item set, we maintain a set $T_{k,\bullet}^i$ of buckets $T_{k,h_{k}(j)}^i$ for all $k\in \{0,\ldots \log |U|\}$ and $h_{k}(j)\in \{0,\ldots, c^2-1\}$. Each such bucket $T_{k,h_k(j)}^i$ contains the sum of the entries hashed to it. This allows us to maintain the contents of $T_{k,h_k(j)}^i$ under dynamic updates.
Note that to support similarity estimation for sets that might have a low cardinality at query time, we must also maintain a bucket set $T_{0,\bullet}^i$ associated to a hash function $h_0$, that will receive all items seen so far for a given set $i$, \ie each of them will be hashed to the bucket $T_{0,h_0(j)}^i$ with probability $2^0=1$ (see line 7 in Algorithm~\ref{alg:preprocess}).

For the interesting values of $k$, \ie $k\in \Theta(\log |A|)$, the number of indexes sampled by $h$ will not exceed some constant $c$. This means that the sampled indexes will be perfectly hashed by $h_k$, \ie the sum contained in $T_{k,h_k(j)}^i$ consists of exactly one item index.
If $k$ is too small (i.e. we sampled too many indexes), $h_k$ has the useful effect of compressing the used space, as $c^2$ counters require at most $O(c^2 \log |U|)$ bits of space.

We can then generate the fingerprint matrix, for instance, by performing a min-hash on the buckets $B_{k,\bullet}^i$ and storing the index of the first non zero bucket.
For a pseudocode of this approach, see Algorithm~\ref{alg:preprocess}. Algorithm~\ref{alg:filter} describes an example candidate generation as per Cohen et al.~\cite{CDFGIMUY01}.

\begin{algorithm}
\begin{algorithmic}[1]
\renewcommand{\algorithmicrequire}{\textbf{Input:}}
\renewcommand{\algorithmicensure}{\textbf{Output:}}
\REQUIRE{Parameter $c \in \mathbb{N}$}
\ENSURE{$T^{(i)}_{k,l}$ with $i \in [n], k \in [0,\ldots,\log m], l \in [c^2]$\\}
Initialization:\\
$s_i = 0$ for all $i \in [n]$\\$T^{(i)}_{k,l} = 0$ for all $i \in [n], k \in [0,\ldots,\log |U|], l \in [c^2]$.\\
$h: [|U|] \rightarrow [M]$ a $2$-universal hash function.\\
$h_1: [M] \rightarrow [c^2]$ another $2$-universal hash function.\\
\STATE {\textbf{On update} $(i,j,v)$:}\\
\begin{ALC@g}
\STATE $k = \lsb(h(j))$\\
\STATE $T^{(i)}_{k, h_1(j)} = T^{(i)}_{k, h_1(j)} + v$\\
\STATE $T^{(i)}_{0, h_1(j)} = T^{(i)}_{0, h_1(j)} + v$\\
\STATE $s_i = s_i + v$
\end{ALC@g}
\end{algorithmic}
\caption{Dynamic stream update (Jaccard)}\label{alg:preprocess}
\end{algorithm}

\begin{algorithm}
\begin{algorithmic}[1]
\renewcommand{\algorithmicrequire}{\textbf{Input:}}
\renewcommand{\algorithmicensure}{\textbf{Output:}}
\REQUIRE{Thresholds $0 < r_1,\alpha < 1$, $T^{(i)}_{k,l}$ from Alg.1 with $k \in \lbrace 0,1,2, \ldots, \log |U| \rbrace$}
\ENSURE{Set of candidate pairs\\}
Initialization:\\
$I = \lbrace 0,\log(1/r_1),2\log(1/r_1), \ldots, \log |U| \rbrace$\\
$H_i$: empty list for $i \in I$.\\
\FOR{$i \in [n]$}
\STATE $s = \ell_0(x^{(i)})$\\
\FOR{$k \in [\log(r_1^2 \cdot \alpha \cdot s), \log(\alpha \cdot s)] \cap I$}
\STATE add $(i,MinHash(T^{(i)}_{k,\bullet}))$ to $H_{k}$

\ENDFOR
\ENDFOR

\RETURN $\lbrace (i,i') \mid \exists k: (i,h),(i',h') \in H_k \text{ and } h = h' \rbrace$
\end{algorithmic}
\caption{Filter candidates (Jaccard)}\label{alg:filter}
\end{algorithm}

\begin{proof}[Proof of Theorem~\ref{thm:minhash_dynamic}]
Fix items sets $A$ and $B$ and let $a,b$ be the corresponding characteristic vectors for the sets $A$ and $B$, respectively. Without loss of generality, assume $|A| \geq |B|$.
Set $\alpha = \frac{\varepsilon^2 \cdot \delta }{600}$.
If $S(A,B) \geq r_1$ then 
$|A|/ |B| \leq 1/r_1$, then $\log(\alpha \cdot|B|) \leq \log(\alpha \cdot|A|)$ and $\log(r_1 \cdot \alpha \cdot |B|) \leq \log(\alpha \cdot|A|)$. Both sets will then enter line 3 of  Algorithm~\ref{alg:filter} for some common values
of $k$, and must exist at least an $H_k$ containing min-hashes from both sets as per line 4.

Let $2^k$ be the largest power of $2$ such that $k\leq \log\left(\alpha\cdot r|A\cup B|\right)\leq \log\left(\alpha\cdot |A\cap B|\right)$. 
Let $U_k$ be a subset of indexes as determined by line 4 of Algorithm~\ref{alg:preprocess} and define $A_k:=U_k\cap A$ and $B_k:= U_k\cap B$.
In expectation, $\mathbb{E}[|A_k\cup B_k|] = |A\cup B|/2^k$.
By Markov's inequality, we have $|A_k\cup B_k| \leq \frac{3}{\delta}\cdot |A\cup B|/2^k \leq \frac{1800}{\varepsilon^2\delta^2\cdot r_1}$ with probability at least $1-\delta/3$. 
By setting the number of buckets in the order of 
\begin{equation}
\label{eq:cbound}
c^2=|A_k\cup B_k|^2 \in O\left(\frac{1}{\varepsilon^4\delta^5\cdot r_1^2}\right),
\end{equation} the elements of $A_k\cup B_k$ will be perfectly hashed by $h_k$ with probability at least $1-\delta/3$ (line 3 of Algorithm~\ref{alg:preprocess}).
Since deleting indexes where both vector entries are zero does not change the Jaccard similarity, the probability that the smallest index in the collection of buckets $T^{(p)}_{k,\bullet}$ is equal to the smallest index in the collection of buckets $T^{(q)}_{k,\bullet}$ is equal to the similarity of $A_k$ and $B_k$. 
Thus we have
$$\mathbb{P}[MinHash(T^{(p)}_{k,\bullet})= MinHash(T^{(q)}_{k,\bullet})]=S(A_k, B_k).$$
If $S(A,B)\geq r_1$ we have by our choice of $\alpha$ and due to the first part of Lemma~\ref{lem:samplesim_general}, $S(A_k,B_k)\geq (1-\varepsilon)\cdot S(A,B)$ with probability $1-\frac{\delta}{3}$.
If $S(A,B)\leq r_2< r_1$, we have due to the second part of Lemma~\ref{lem:samplesim_general} $S(A_k, B_k)\leq \dfrac{6\cdot S(A,B)}{\delta(1-(\varepsilon/5) \cdot \sqrt{2r_1})} \leq \frac{6r_2}{\delta(1-(\varepsilon/5) \cdot \sqrt{2r_1})}$ with probability $1-\frac{\delta}{3}$.

Conditioning on all events gives us a $(r_1,r_2,(1-\varepsilon)r_1,6r_2 /(\delta(1-\varepsilon/5 \sqrt{2r_1}))$-sensitive LSH with probability $1-\delta$.

To bound the space requirement, observe that for each of the $n$ item sets we have $\log |U|$ collections $T^{(p)}_{k,\bullet}$ of $c^2\in O\left(\frac{1}{\varepsilon^4\delta^5\cdot r_1^2}\right)$ buckets due to Equation~\ref{eq:cbound}. Each bucket contains a sum that uses at most $\log |U|$ bits.
The space required for each hash function is at most $\log |U|$ due to Dietzfelbinger~\cite{Die96}.
\end{proof}



The parameters in Theorem~\ref{thm:minhash_dynamic} can be chosen such that we are able to use Algorithm \ref{alg:preprocess} and Algorithm \ref{alg:filter} similar to the min-hashing technique in the non-dynamic scenario.
This also means that we can use similar tricks to amplify the probability of selecting high similar items in Algorithm \ref{alg:filter} and lower the probability in case of a small similarity as long as $(1-\varepsilon)r_1>\frac{6r_2}{\delta(1-\varepsilon/5)\sqrt{r_1}}$, see also Leskovec et al.~\cite{LRU14}. 
Let $\ell,m \in \mathbb{N}$. Then we repeat the hashing part of Algorithm~\ref{alg:filter} $\ell$ times and only add a pair to the output set iff all $\ell$ hash values are equal. This procedure is repeated $m$ times and the final output set contains all pairs which appear at least once in an output set of the $m$ repetitions. The probability that a pair with similarity $s$ is in the output set is $1-(1-p^\ell)^m$ with $p \geq (1-\delta)(1 - \varepsilon)s$ if $s > r_1$ and $p \leq 6 s/(\delta(1-\varepsilon/5 \sqrt{r_2})$ if otherwise.

\section{Experimental Evaluation}
\label{sec:experiments}
In this section we evaluate the practical performance of the algorithm given in Section~\ref{sec:filter}.
Our aim is two-fold: First, we want to show that the running time of our algorithm is competitive with more conventional min-hashing algorithms. 
For our use-case, \ie dynamic streams, we are not aware of any competitors in literature. 
Nevertheless, it is important to demonstrate the algorithm's viability, as in many cases a system might not even support a dynamic streaming environment: we show a performance comparison in Section~\ref{subsec:running_time}.
To cover all ranges of user profiles, we use a synthetic benchmark described below.

Our second goal is to be able to measure the quality of the algorithm's output. We deem our filtering mechanism to be successful if it finds most of the user pairs with high similarity, while performing a good level of filtering, returning as candidates few user pairs with low similarity. 
Furthermore, Theorem~\ref{thm:minhash_dynamic} guarantees us a reasonable approximation to the similarity of each pair, though it is unclear whether this still holds for all pairs simultaneously, especially for small bucket sizes.  We are satisfied if our approximate computation based on sketches does not lead to high deviation with respect to exact similarities. 
As a typical candidate from practice, we consider user profiles containing a users recently preferred artists from Last.FM. 

\noindent
{\bf Implementation details:} 
We implemented Algorithm~\ref{alg:preprocess}, Algorithm~\ref{alg:filter}, as well as other hash routines in C++ and compiled the code with GCC version 4.8.4 and optimization level 3.
Compared to the description of Algorithm~\ref{alg:filter}, which has 5 parameters (error $\varepsilon$, failure probability $\delta$, lower bound for desirable similarities $r_1$, upper bound for undesirable similarities $r_2$, and granularity of the sampling given by $M$), our implementation has only two parameters: (1) the inverse sampling rate $\alpha$ and (2) the number of buckets $c^2$. Recall that a higher inverse sampling rate $\alpha$ means selecting higher values of $k$ in Algorithm~\ref{alg:filter}, line 5, where an increasing $k$ is associated to a decreasing sampling probability $2^{-k}$ of a bucket $T_{k,h_k(i)}$.

The choice of $c^2$ influences the possible combinations of $\varepsilon$, $\delta$, and $r_1$, see Equation~\ref{eq:cbound} for an upper bound on $c^2$.
The cardinality based filtering of Algorithm~\ref{alg:filter} is influenced by the choice of $\alpha$. 


As a rule of thumb, $r_2$ is roughly of the order $r_1^{1.5}$.
For example, if we aim to retain all pairs of similarity at least $\frac{1}{4}$, we can filter out pairs with similarity less than $\frac{1}{8}$. Pairs with an intermediate similarity, i.e. a similarity within the interval $[\frac{1}{8},\frac{1}{4}]$, may or may not be detected. We view this as a minor restriction as it is rarely important for these thresholds to be sharp.

Lastly, we implemented Dietzfelbinger's multiply-add-shift method to generate $2$-wise independent hash functions, where $a$ is a random non-negative odd integer, $b$ a random non-negative integer, and for a given $M$ the shift is set to $w-\log(M)$, where $w$ is the word size (32 bits in our implementation).
All hash functions used in the implementation of both Algorithm~\ref{alg:preprocess}, that is the functions $h$, $h_1$ and the hash functions used for implementing the MinHash scheme, with amplification parameters $\ell$ (functions in one band) and $m$ (number of bands) at line 6 of Algorithm~\ref{alg:filter}), are  $2$-wise independent hash functions, and were generated independently, \ie we did not reuse them for subsequent experiments.  
	Otherwise the implementation follows that of Algorithms~\ref{alg:preprocess} and~\ref{alg:filter} with various choices of parameters.

All computations were performed on a 2.7 GHz Intel Core i7 machine with 8 MB shared L3 Cache and 16 GB main memory. 
Each run was repeated 10 times. 


\noindent
\textbf{Synthetic Dataset}  
To accurately measure the distortion on large datasets, for varying feature spaces,
we used the synthetic benchmark by Cohen et al.~\cite{CDFGIMUY01}. 
Here we are given a large binary data-matrix consisting of $10,000$ rows and either $10,000$, $100,000$ or $1,000,000$ columns. 
The rows corresponded to item sets and the columns to items, \ie we compared the similarities of rows.
Since large binary data sets encountered in practical applications are sparse, the number of non-zero entries of each row was between $1\%$ to $5\%$ chosen uniformly at random.
Further, for every $100$th row, we added an additional row with higher Jaccard similarity in the range of $\{ (0.35, 0.45), (0.45, 0.55), (0.55, 0.65),$ $ (0.65, 0.75), (0.75, 0.85),$ $(0.85, 0.95)\}$. 

To obtain such a pair, we copied the preceding row (which was again uniformly chosen at random) and uniformly at random flipped an appropriate number of bits, \eg for $10,000$ items, row sparsity of $5\%$, and similarity range $(0.45,$  $0.55)$ we deleted an item contained in row $i$ with probability $1/3$ and added a new item with probability $\frac{1}{19}\cdot \frac{1}{3}=\frac{1}{57}$.
In the insertion-only case, the stream consists of the sequence of $1$-entries of each row. We introduced deletions by randomly removing any non-zero entry immediately after insertion with probability $\frac{1}{10}$.

\noindent \textbf{Last.FM Dataset:} For an evaluation of our algorithm on real data we considered a dataset from~\cite{anagnostopoulos2015learning} 
containing temporal data from the popular online (social) music recommendation system Last.fm.
Users build a profile in multiple ways: listening to their personal music collection with a music player app with the Last.fm \emph{Audioscrobbler} plugin, or by listening to the Last.fm internet radio service, either with the official client app, or with the embedded player. Radio stations consist of uninterrupted audio streams of individual tracks based on the user's personal profile, its ``musical neighbors'' (\ie similar users identified by the platform), or the user's ``friends''. All songs played are added to a log from which personalized top charts and musical recommendations are calculated, using a collaborative filtering algorithm. This automatic track logging is called \emph{scrobbling}. 

Our dataset contains the full ``scrobbled'' listening history of a set of $44,154$ users, covering a period of 5-years from May 2009 to May 2014, containing 721M listening events and around 4.6M unique tracks, where each track is labeled with a score for a set of 700 music genres. To obtain a more granular feature space, we decided to map each track to the corresponding artist. To this end we queried the MusicBrainz DB
to obtain artist information for each of the unique tracks (total of $1.2M$ artists). We then processed the listening histories of each user $i$ in chronological order to produce our event stream, emitting a triple \texttt{(i,j,+1)} after a user has listened to at least 5 songs by an artist $j$, and emitting a triple \texttt{(i,j,-1)} when no track from artist $j$ is listened by $i$ for a period of 6 months (expiration threshold). 
The initial threshold is mainly intended to mitigate the ``recommendation effect'': being Last.fm a recommendation system, some portions of the listening histories might in fact be driven by recommendation sessions, where diverse artists are suggested by the system based on the user's interests (\ie not explicitely chosen by him), and are likely to lead to cascades of item losses after the expiration time. Like most real world datasets that link users to bought/adopted items, this dataset is very sparse. 
To overcome this problem, we considered only users having at least $0.5\%$-dense profiles on average, obtaining a final set of $n=15K$ users (sets), $|U|=380K$ (items) and a stream length of $6.2M$.
Table~\ref{tab:lastfm_dist} shows the distribution of exact similarity values for all pairs the Last.fm dataset.

\begin{table}[htbp]
\begin{center}
\begin{tabular}{c | c | c | c | c }
Similarity & 0.0 & 0.05 & 0.1 & 0.15  \\

Number of pairs & 60432710 & 37947485 & 12031795 & 1938117  \\
\hline
Similarity & 0.2 & 0.25 & 0.3 & 0.35  \\

Number of pairs & 164246 & 7855 & 266 & 13 \\
\hline
Similarity & 0.4 & 0.45 & 0.5 & $\geq$ 0.55 \\
Number of pairs & 10 & 1 & 3 & 0 
\end{tabular}
\caption{Distribution of exact similarity values for pairs in the Last.fm dataset}
\label{tab:lastfm_dist}
\end{center}
\end{table}

\subsection{Performance Evaluation}
\label{subsec:running_time}
We evaluated the running time of our algorithm using the synthetic dataset, to understand its performance with respect to various dataset sizes, in  two different scenarios, an \emph{insertion-only stream}, and a \emph{fully dynamic} stream, both obtained from our synthetic dataset.
As a comparative benchmark, we compare our approach with an online implementation of a ``vanilla'' LSH scheme (later \emph{Vanilla-MH}), where profile sketches are computed online using $2$-wise independent hash functions (that is also our signature scheme). 


\noindent We tested two versions of our algorithm. The first version henceforth called \textit{DynSymSearch Proactive} (or simply \textit{Proactive}) maintains a set of fingerprints online, with every update (that is, after line 7 of Alg.~\ref{alg:preprocess}, reflecting the most recent change from the stream).
The second version henceforth called \textit{DynSymSearch} maintains the sketches of Alg.~\ref{alg:preprocess} and computes fingerprints only at query time.


The choice of the first or the second implementation depends on the use case, with a trade-off between query responsiveness and additional space required for computing and storing the signatures of sets. 

Let us now focus on the algorithms that update signatures online.
When inserting item $j$ added to set $i$, both \textit{Proactive} and \emph{Vanilla-MH} behave in the same way. When an element is added, all hash-functions are evaluated on the new element, and updated in case such value is the new minimum. 
Let $k=lsb(h(j)$.
In case of deletions, both will have to recompute signatures, yet while \emph{Vanilla-MH} has to do so for the full user profile, \textit{Proactive} has to recompute signatures only for the two compressed bucket sets $T_{k,\bullet}^i$  and $T_{0,\bullet}^i$. 
A further optimization that we implemented in \textit{Proactive}, is the \emph{selective recomputation} of signatures in case of deletions. In case of deletions of an item $j$, we recompute a set of signatures for $T_{k}^{(i)}$ only if the bucket is \emph{sensitive}, \ie its corresponding set cardinality and similarity threshold are such that $k$ is the range specified by line 5 of Algorithm~\ref{alg:filter}. This allows to ignore many costly recomputations.

Now we can move on to comparing the three on the various settings.
We set $\ell=5$, $m=40$ as amplifying parameters for signatures of all algorithms, and further set $r_1=0.5$ for our two algorithms. The choice of $\ell$,$m$ is not extremely important: indeed for the sake of runtime comparison all algorithms should only share the same ``hashing-related'' overhead. 
Average running times of 10 independent realizations of each algorithm are plotted in Figure~\ref{fig:running_time} where we study the impact of the parameters.

The running time of our algorithms is influenced by their parameters to different extents. In particular, the number of buckets $c^2$ has impact on both our algorithms (especially for \textit{Proactive}) as it directly implies more hash function evaluation for fingerprints.

For the insertion-only stream (Figure~\ref{fig:running_time_insertions_only}),  we see that the performances of the three algorithms are somewhat comparable, which is expected, considering that \emph{Vanilla-MH} is to some extent naturally contained in both versions of our algorithm. 
In \textit{DynSymSearch} they are computed only at query time on the sensitive sketches, rendering it the fastest option for this scenario. 

When considering deletions, things change dramatically. As can be seen in Figure~\ref{fig:running_time_dynamic}, deletions represent a problem for both \emph{Vanilla-MH} and \textit{Proactive}: as the fingerprint computation is not reversible, after a deletion they must all be consistently recomputed. However, our algorithm is less affected: thanks to its compression and cardinality-based bucketing system, the updates are, to some extent, more local, as they impact only the sensitive buckets. We note that \textit{Proactive} has some values of $\alpha$ where the running time increases: these values allow for a wider range of buckets to become sensitive as long as the set cardinalities vary with the stream, implying more signature recomputations when each $k$ becomes queryable.
Overall, \textit{DynSymSearch} is consistently faster then the other two options, while the performance of \emph{Vanilla-MH} is very poor, taking from 12 to 100 times more time than \textit{Proactive}, for $d=1M$.
We also remark that in the implementation of \emph{Vanilla-MH}, we are forced to store the entire data set in order to deal with deletions, to be able to recompute the fingerprints. This requirement is indeed not feasible in many settings.
Furthermore, even for these comparatively sparse data sets, our algorithm has significant space savings.

\begin{figure}
 \centering
 \begin{subfigure}[]{\columnwidth}
 \centering
\includegraphics[width=\columnwidth]{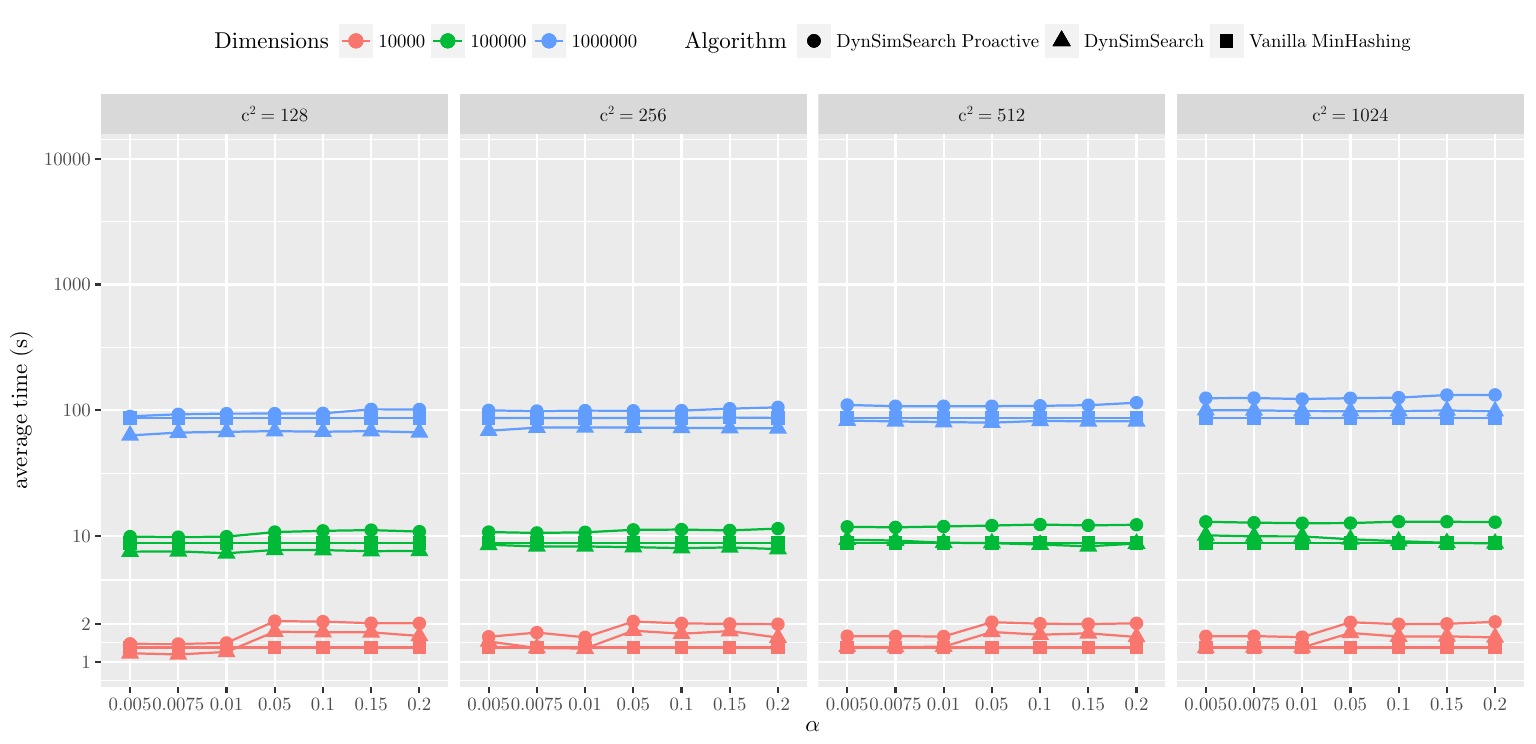}
 \caption{Insertion only stream}
 \label{fig:running_time_insertions_only}
   \end{subfigure}
    \vfill
     \begin{subfigure}[]{\columnwidth}
 \centering
\includegraphics[width=\columnwidth]{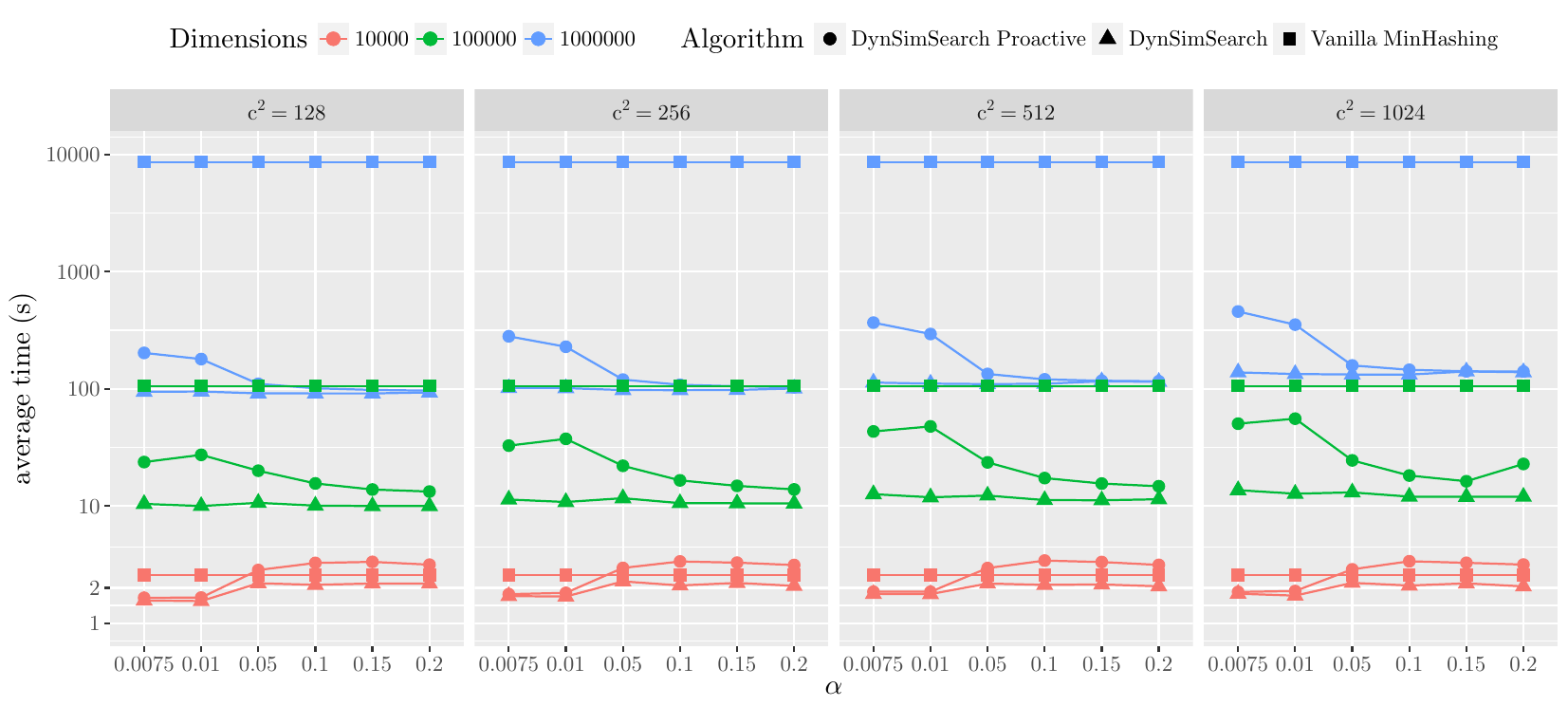}
 \caption{Fully dynamic stream}
 \label{fig:running_time_dynamic}
   \end{subfigure}
      \caption{[\textbf{Synthetic}] Running time of our algorithm compared to a 2-wise MinHash based LSH implementation, in insertion-only or fully dynamic setting for different values of $d$. $y$-axes are in log-scale. The summary running times are the mean values of 10 repetitions.}
\label{fig:running_time}
\end{figure}   

\textbf{Quality of approximation:}
We now move to examine the quality of approximation of our algorithm (which is the same for both implementations), on the synthetic dataset, as a function of our two main parameters, $\alpha$ and $c^2$.   
Concerning $\alpha$, there are two opposite cases. If the inverse sampling rate is too low, we might have chosen set representative buckets with many samples: this means high chance of collisions which decreases the approximation ratio. On the other hand if it is too high, the selected set of items might not be sensitive.
A higher bucket size instead, always means less collisions, for an increased space occupation of the sketches.

Figure~\ref{fig:deviation_synthetic} shows the values of the average squared deviation of the sketched similarities obtained with our algorithm, and their exact Jaccard similarity, on the synthetic dataset, for different value of $d$, and various values of the parameters $\alpha$ and $c^2$.

The goodness of a given $\alpha$ depends on the similarity of a pair in question.
We show separate plots for high and low similarity pairs, that is pairs with Jaccard similarity respectively below and above $0.2$. Their behavior is affected in a different way.
First, low pairs tend to have higher average squared deviation than high pairs, this is expected as out sketches can better approximate high similarity pairs.
Also, for both kind of pairs the distortion decreases with increasing $c^2$, independently of $\alpha$ as the number of collisions decrease monotonically. All deviations reach almost zero already at $\alpha=0.05$ for all bucket sizes. For $\alpha$  above $0.1$ we see that the deviation of high similarity pairs depart from the others, and especially for higher dimensional datasets tend to be slightly more distorted.
Except for the lowest number of buckets, the average total deviation for these parameters was always below $0.1$ and further decreased reaching to zero for larger bucket sizes. 
We note that these values of $c^2$ are below the theoretical bounds of Theorem~\ref{thm:minhash_dynamic}, while having little to acceptable deviation for appropriately chosen values of $\alpha$.

\begin{figure}
 \centering
 \begin{subfigure}[]{\columnwidth}
 \centering
\includegraphics[width=\columnwidth]{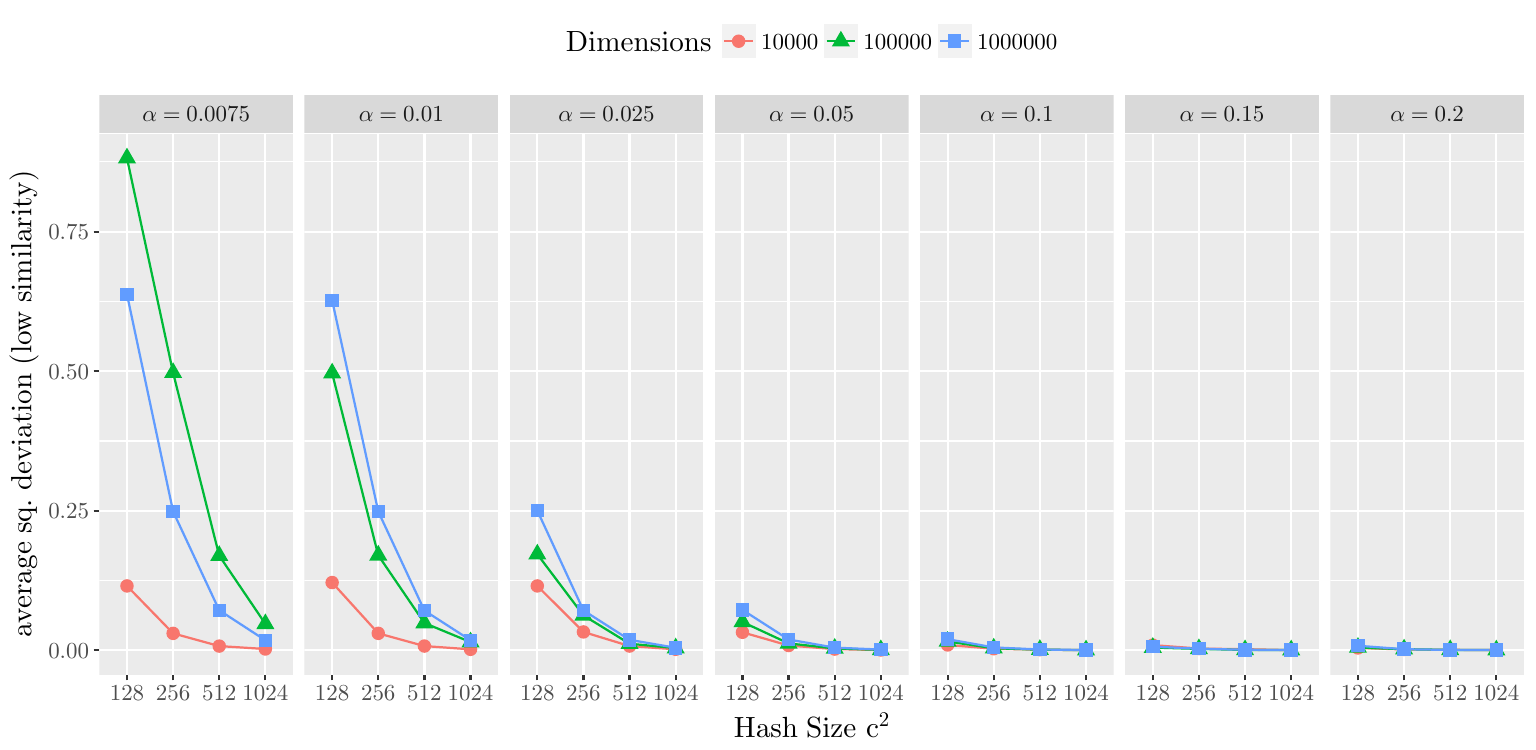}
 \caption{Low similarity pairs}
   \end{subfigure}
    \vfill
     \begin{subfigure}[]{\columnwidth}
 \centering
\includegraphics[width=\columnwidth]{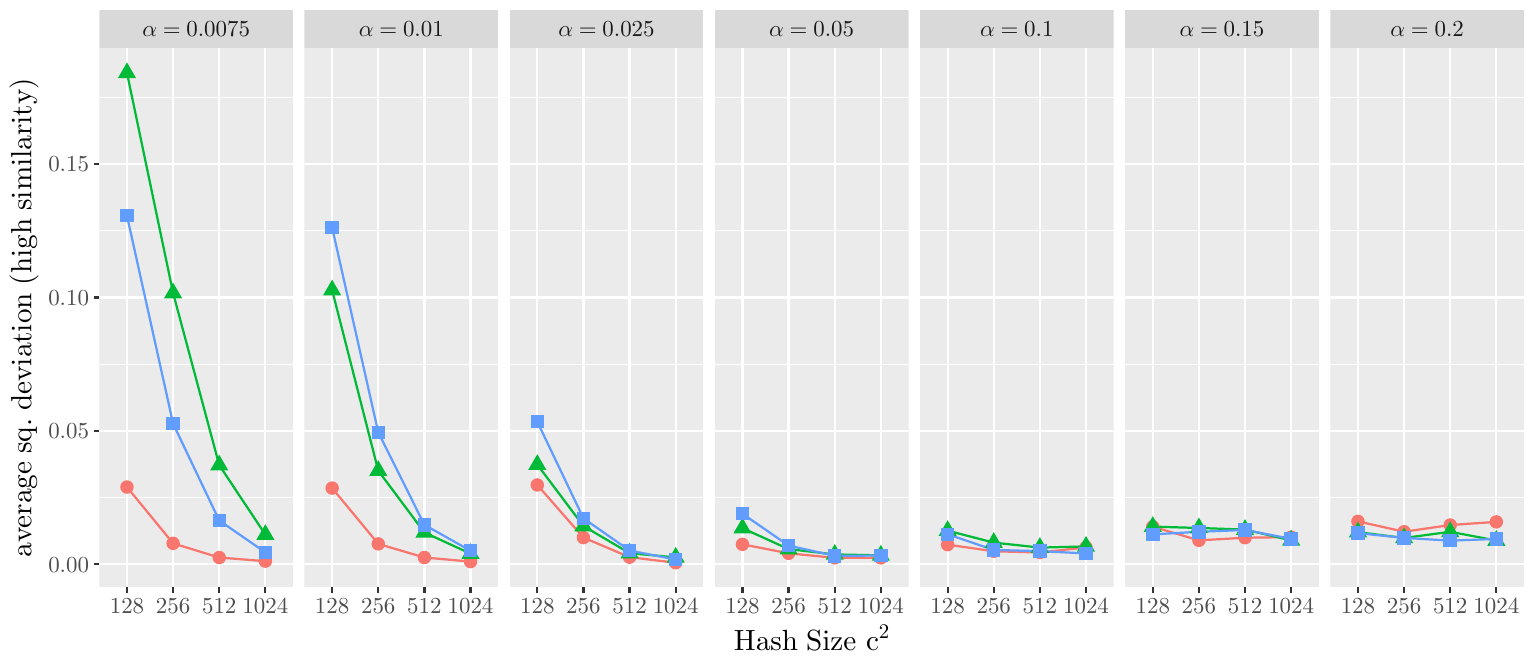}
\caption{High similarity pairs}
   \end{subfigure}
      \caption{[\textbf{Synthetic}] Average squared deviation for high similarity ($J \geq 0.2$) and low similarity ($J < 0.2$) pairs in the synthetic dataset, for various parameter choices}
\label{fig:deviation_synthetic}
\end{figure}

\subsection{Analysis of the Last.fm Dataset}
\label{sec:analysis_lastfm}
A realistic context like the one of Last.FM dataset, offers a valid playground to explore the performance of our similarity search. 
We use the locality sensitive hashing approaches as a recommendation device. 
We also compute a visualization of the most related user pairs, which illustrates an application of the  sketching techniques from Section~\ref{sec:est} to implicitely store an approximate distance matrix in small space.
We note that the data is very sparse. Since the  Jaccard index is highly sensitive to the support of the vectors, using it for this type of recommendation is more appropriate compared to other similarity measures such as Hamming, or cosine similarity. 


In Algorithm 2, we fixed $r_1 = 0.25$, therefore for each set $A$ added a min-hash value to $H_k$ with $k \in  [s - 2, s] \cap I$, with $I=\{0, log_2\frac{1}{r_1}, \dots\}= \{0, 2, 4, \dots\ , log\rvert U\rvert \}$, where $s$ is the actual cardinality of $A$ at the time of the query (which we perform at the end of the stream). 
At the output of the filtering phase, we evaluate the similarity between users of a candidate pair using $k = \log(\alpha \cdot r_1 \cdot \max(\rvert A \lvert,\rvert B \lvert)$ and output $S(A_k, B_k)$. Note that this choice of $k$ satisfies the first condition from Lemma~\ref{lem:samplesim_general}. 
Note that this dataset, as witnessed by the huge presence of very low similarity pairs (see Table~\ref{tab:lastfm_dist}), and very few pairs with higher similarity, is a challenge for any LSH scheme, as providing a good filtering behavior with low similarity thresholds requires many hash functions.

We performed multiple experiments in order to choose good parameters of $\ell$ and $m$ to achieve a good filtering. We set a threshold on the maximum number of hash functions to use to 1600 hash functions. Then we also set a threshold on the maximum fraction of pairs that we accept to be reported as candidate pairs, to $10\%$. Then we tested a number of combinations of $\ell$ and $m$ that are compatible with the similarity threshold $r$ and meet our constraints, and report them in Table~\ref{tab:ab_choices}.
The combination of $\ell=5$ and $m=300$ shows the lowest number of false negatives, and achieves a very good filtering, reporting only as low as 3.6\% of pairs. We choose these values as amplification parameters for the filtering phase, and are fixed for all the experiments on this dataset.

\begin{table}[ht!]
  \centering
  \begin{tabular} {| r | r | r | r |}
    
    \hline  
    \textbf{$\ell$} & \textbf{$m$} & \textbf{\% candidate pairs} & \textbf{False negatives}\\
    \hline
    \multirow{4}{*}{}
 4 & 400 & 0.099	&1393 \\ 
 \textbf{5} & 50 & 0.016 & 5171 \\ 
   & 150 & 0.028 &2406  \\ 
  & \textbf{300} & \textbf{0.036} & \textbf{781}  \\
  & 320 & 0.029	& 1587  \\
 6 & 200 & 0.0187 & 5415  \\
\hline    
  \end{tabular}

\caption{Fraction of pairs reported as candidates vs best number of false negatives given by our algorithm for various choices of $a$ and $b$}
\label{tab:ab_choices}
\end{table}

Figure~\ref{fig:lastfm_deviation} shows average squared deviation values of the sketched similiarities obtained with our algorithm and their exact Jaccard similarity, as function of $\alpha$ and $c^2$.
Like in Figure~\ref{fig:deviation_synthetic}, we show separate curves for pairs with Jaccard similarity below $0.2$ (green curve) and high pairs (red curve).
The same considerations made for the synthetic dataset hold, while we note that, for this dataset, the approximations of high similar pairs for very low bucket sizes appear slightly worse, possibly because indeed the majority of them have a similarity value is closer to the threshold, with respect to the synthetic dataset. However, for appropriate values of the parameters, all deviations tend to zero. 
Figure~\label{fig:lastfm_ac_rec_output} shows other information regarding the detection performance of our filtering scheme.
Recall that the sensitivity of our scheme is defining using Indyk and Motwani~\cite{InM98} kind of sensitivity, that is characterized by two different thresholds $r_2 < r_1$ (and corresponding regimes, with different approximation bounds as per Theorem~\ref{thm:minhash_dynamic}). As a rule of thumb, $r_2$ is roughly of the order $r_1^{1.5}$, so we tolerate to report pairs with similarity above $r_2=0.125$ , and consider this range as \emph{true positives} (\textbf{TP}), \emph{true negatives} (\textbf{TN}) pairs below $r_2$ that are correctly not reported. Conversely, pairs below $r_2$ that are reported as candidates by our algorithm are \emph{false positives} (\textbf{FP}), and we consider \emph{false negatives} (\textbf{FN}) pairs that are above the real threshold $r_1=0.25$ but were not reported. Figure~\ref{fig:lastfm_ac_rec_output} shows values of $Accuracy=\frac{(\textbf{TP+TN})}{(\textbf{TP+TN+FP+FN})}$, $Recall=\frac{(\textbf{TP})}{(\textbf{TP+FN})}$ and fraction of candidate pairs reported. We can notice that the recall is approximately $1$ for all values of the parameters. Accuracy instead, increases for increasing $c^2$, as expected, and also for increasing $\alpha$, until it deteriorates for very high values, like it was for the high similarity pairs in Figure~\ref{fig:lastfm_deviation}. We notice that we get filtering above 90\% starting from $\alpha=0.075$, for 512 buckets.
Lastly, in Figure~\ref{fig:lastfm_running_time} we see that we achieve very small running times for from $\alpha=0.075$, as a consequence of the filtering. 
We remark that these plots show the performances of the filtering algorithm alone without any further pruning step. Yet, as reported by our very low deviation from actual similarities, we note that when completely avoiding false negatives is of primary concern, one can decide to choose a lower $r_1$ (and/or a different $a, b$ combination) to retain more pairs in the candidate selection phase, and then perform another linear filtering using the accurate estimation given by our sketches.

\begin{figure}[t]
 \centering
 \begin{subfigure}[]{\columnwidth}
 \centering
\includegraphics[width=\columnwidth]{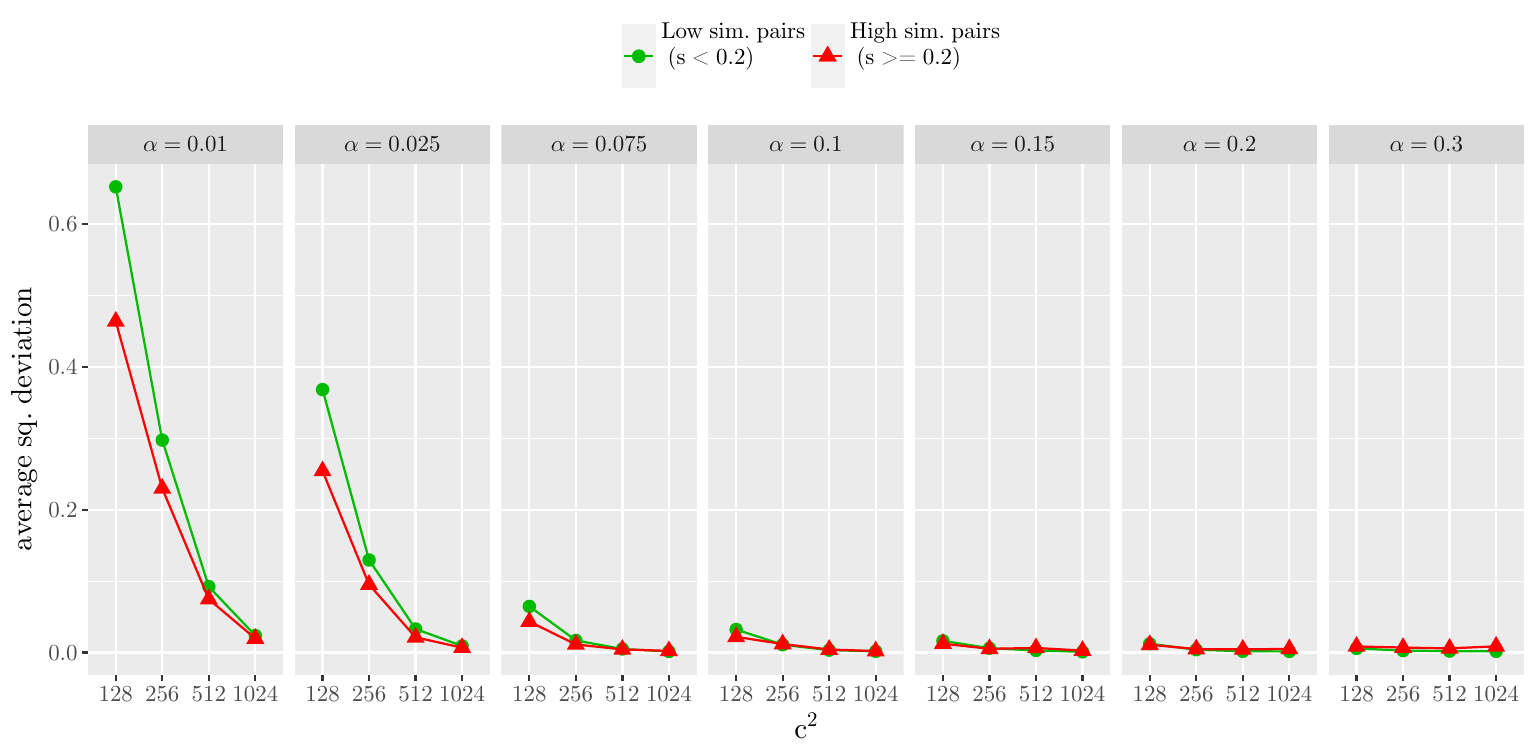}
 \caption{Average squared deviation}
\label{fig:lastfm_deviation}
   \end{subfigure}
    \vfill
     \begin{subfigure}[]{\columnwidth}
 \centering
\includegraphics[width=\columnwidth]{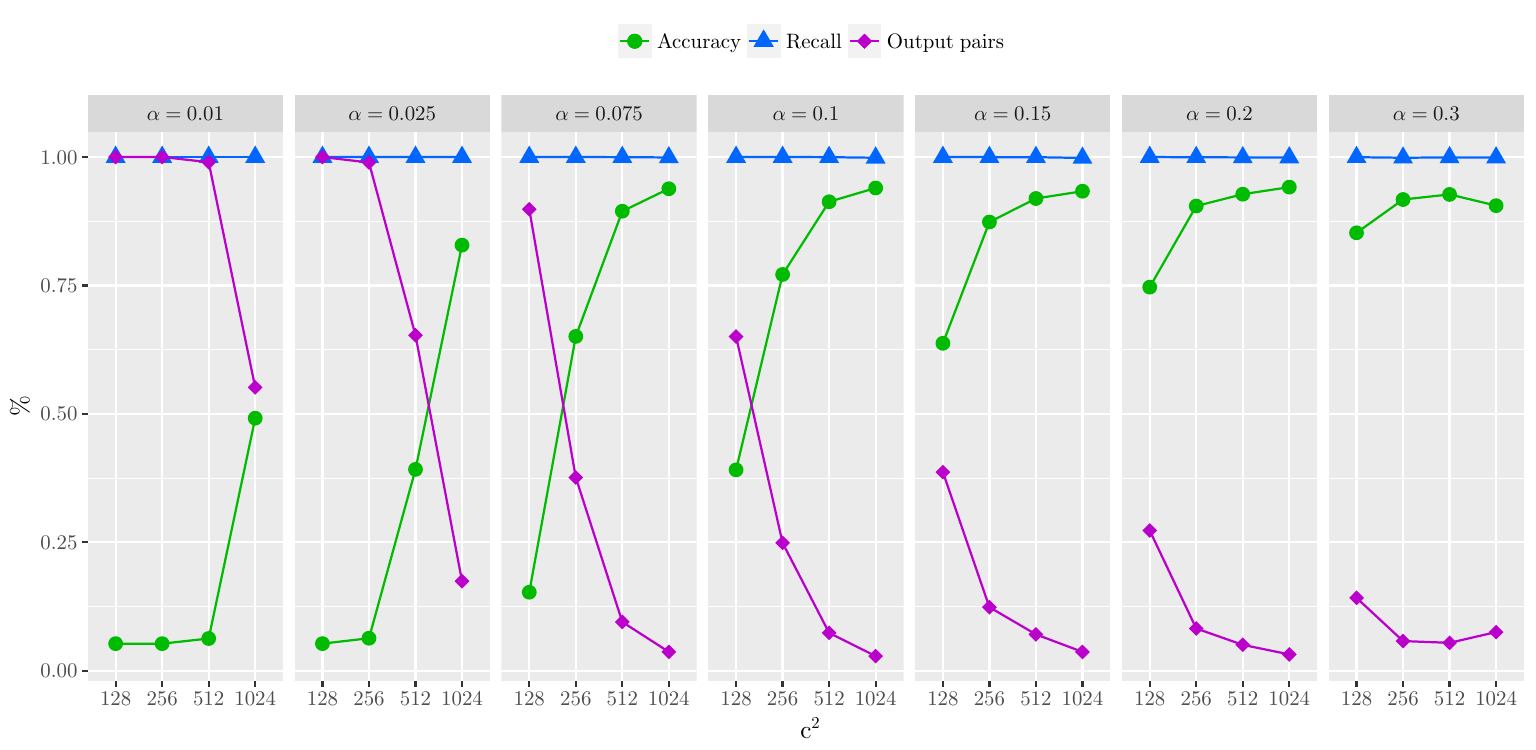}
 \caption{Accuracy, Recall and Fraction of output pairs}
 \label{fig:lastfm_ac_rec_output}
   \end{subfigure}
   \vfill
   \begin{subfigure}[]{\columnwidth}
 \centering
\includegraphics[width=\columnwidth]{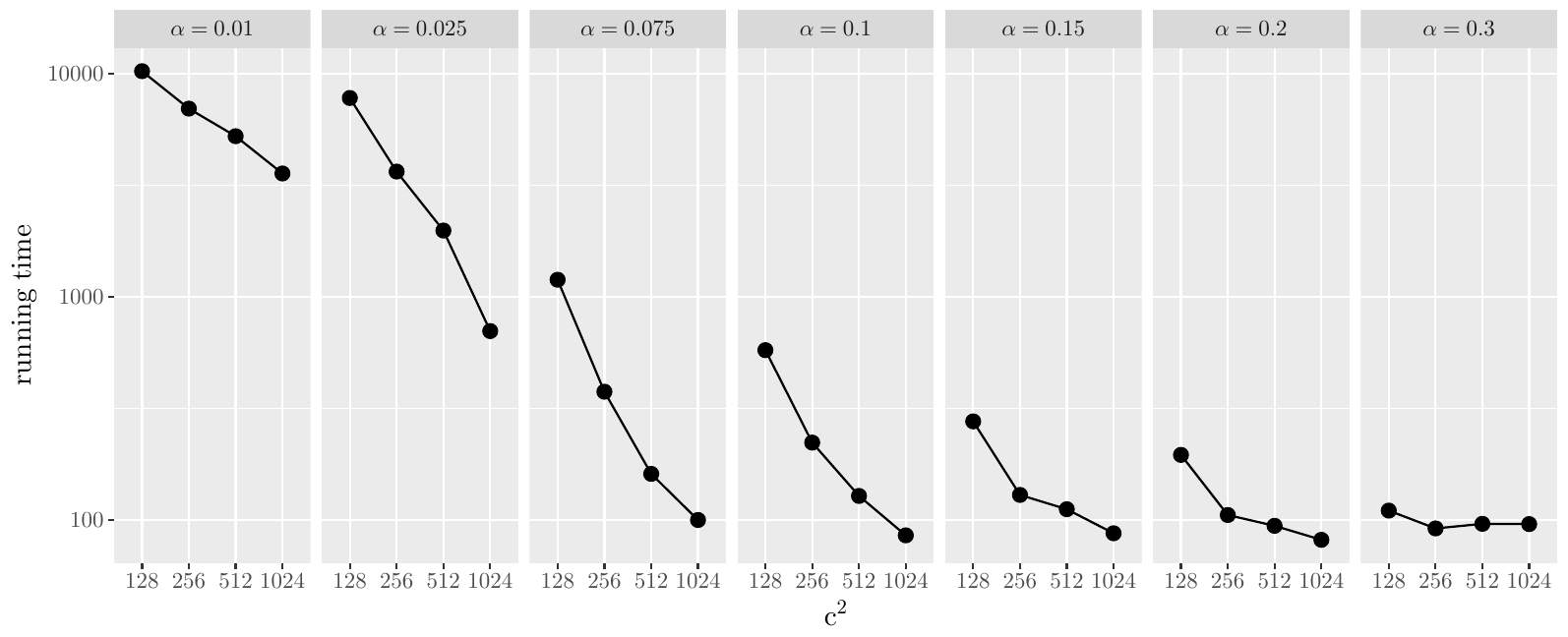}
 \caption{Running time (log-scale)}
 \label{fig:lastfm_running_time}
   \end{subfigure}

      \caption{[\textbf{Last.fm}] Approximation quality, Accuracy, Recall, Fraction of pairs found and running time, for various combinations of the parameters}
  
\end{figure}

\subsubsection{Visualizing top similar users}
\label{sec:visualizing}
We conclude showing a visualization of the most similar Last.fm users found by \emph{DynSymSearch}.
For a predefined order of the elements in $U$, that is, our collection of music artists, we can view user profiles as their characteristic binary vectors, where an entry is $1$ at a given time if a given user has recently listened to the corresponding artist. 
Given the high dimensionality of $U$, it is very hard to find a way to make sense of such similarities. 
We have taken various steps to reach the following two objectives: i) find a lower dimensional representation (ideally 2D points) of the user profiles that can mostly retain their Jaccard similarities, and ii) enrich such points with lower resolution information that helps to visually distinguish similar pairs without recurring to artist annotation.

\begin{figure}

\centering
\includegraphics[width=\columnwidth]{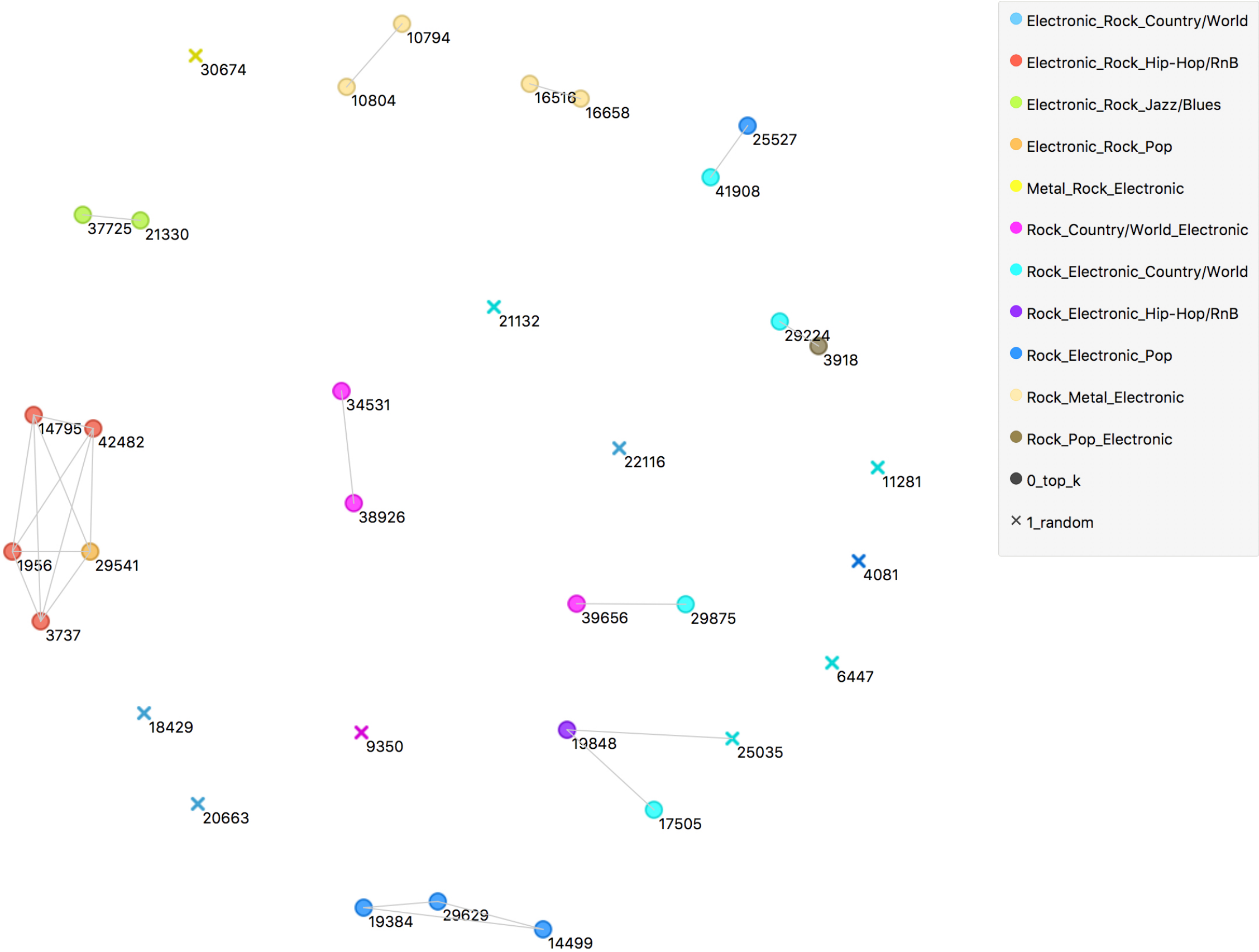}
        \caption{$2D$ embedding of the profiles of 30 users (\emph{top} and \emph{random}), obtained running a MDS algorithm on the characteristic vectors. Circles are used to represent users belonging to top similar pairs ($\SIM \geq 0.4$), and other 10 users were selected at random, and are marked with a cross symbol. Colors represent combinations of the top 3 user genres scores (see Section~\ref{sec:visualizing} for further details).}
        \label{fig:jac_mds}
\end{figure}

Our input is a set of characteristic vectors, representing profiles of a set $S$ of 34 users, 24 of which form the top 14 similar pairs \ie pairs with similarity above 0.4,see Table~\ref{tab:lastfm_dist}, and other 10 users selected at random. We refer to the former users as the \emph{top k} users, and call the latter users \emph{random}.

For implementing step i) we resort to Multidimensional Scaling 
(MDS)~\cite{wickelmaier2003introduction}, a technique that takes in input a matrix of pairwise distances (notably Euclidean and Jaccard, among others) of an input set of objects, and finds an $N$-dimensional projection of points, such that the between-object distances are preserved as well as possible. Each object is then assigned coordinates in each of the $N$ dimensions. 
We used a Python implementation of MDS from the Orange Data Mining library~\cite{JMLR:demsar13a}, where we set $N=2$ and input a Jaccard distance matrix computed on all pairs of our 34 user characteristic vectors.

As per step ii) we used genre information from the original dataset in the form of a vector of scores for the music genres \emph{Rock, Pop, Electronic,Metal,Hip-Hop/RnB, Jazz/Blues} and \emph{Country/World}, where each entry is normalized so that their sum adds to 1.  For each artist $a$ appearing in some user profile, we derived a score vector computing the normalized sum of all score vectors of tracks authored by him, present in our dataset, and then in turn used the same mechanism for deriving a score for a user listening to a set of artists, determining a $7$-dimensional vector, or \emph{genre-based profile} for each user. 

Figure~\ref{fig:jac_mds} depicts a result of the combinations of both steps: a 2-dimensional MDS visualization of $S$ computed using artist-based Jaccard similarity, annotated with colors reflecting the first 3 entries by score as per the corresponding users genres-based profiles. Also, edges show pairs for which the Jaccard similarity is above the threshold 0.1 (note that more than 14 edges are reported, as some users are involved in mildly-similar pairs with other users from \emph{top k}, yet with similarity lower than 0.4). We notice that the majority of users have \emph{Rock} or \emph{Electronic} among their main genres, this is a characteristic of the dataset.

Overall, a clustered structure becomes apparent when considering both distance and genre-based colors, (also, we see that \emph{random} users --- marked with a cross symbol --- are mostly spread out and are not involved in any pairs). Some pairs of \emph{top} users have different colors: this possibly means that their intersection involves a subset of such genres, which is quite natural.

\begin{figure}[ht]
 \centering
 \begin{subfigure}[]{\columnwidth}
 \centering
\includegraphics[width=0.7\columnwidth]{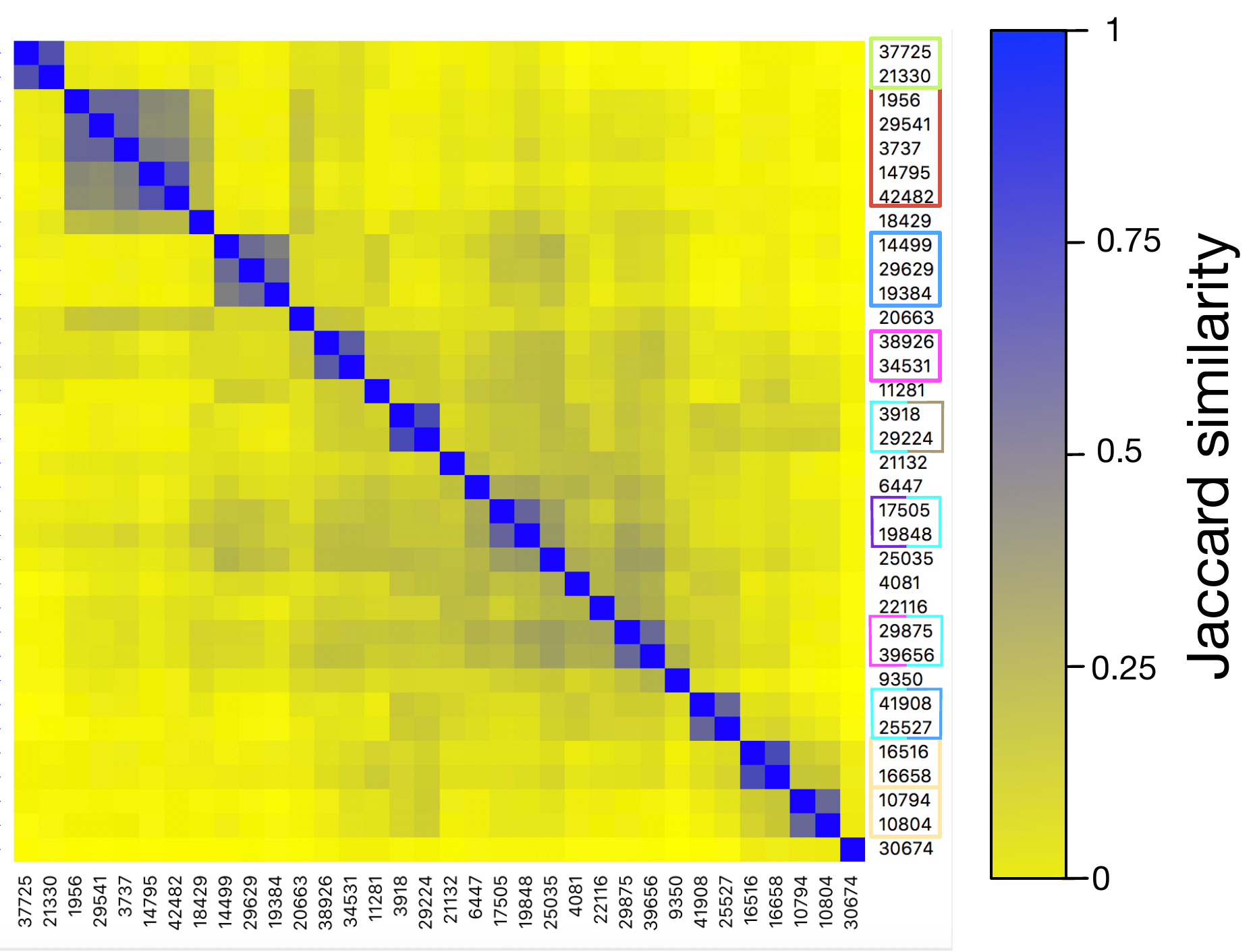}
 \caption{Clustered Jaccard similarity matrix (users profiles)}
\label{fig:clusters_jaccard}
   \end{subfigure}
    \hfill
     \begin{subfigure}[]{\columnwidth}
 \centering
\includegraphics[width=0.7\columnwidth]{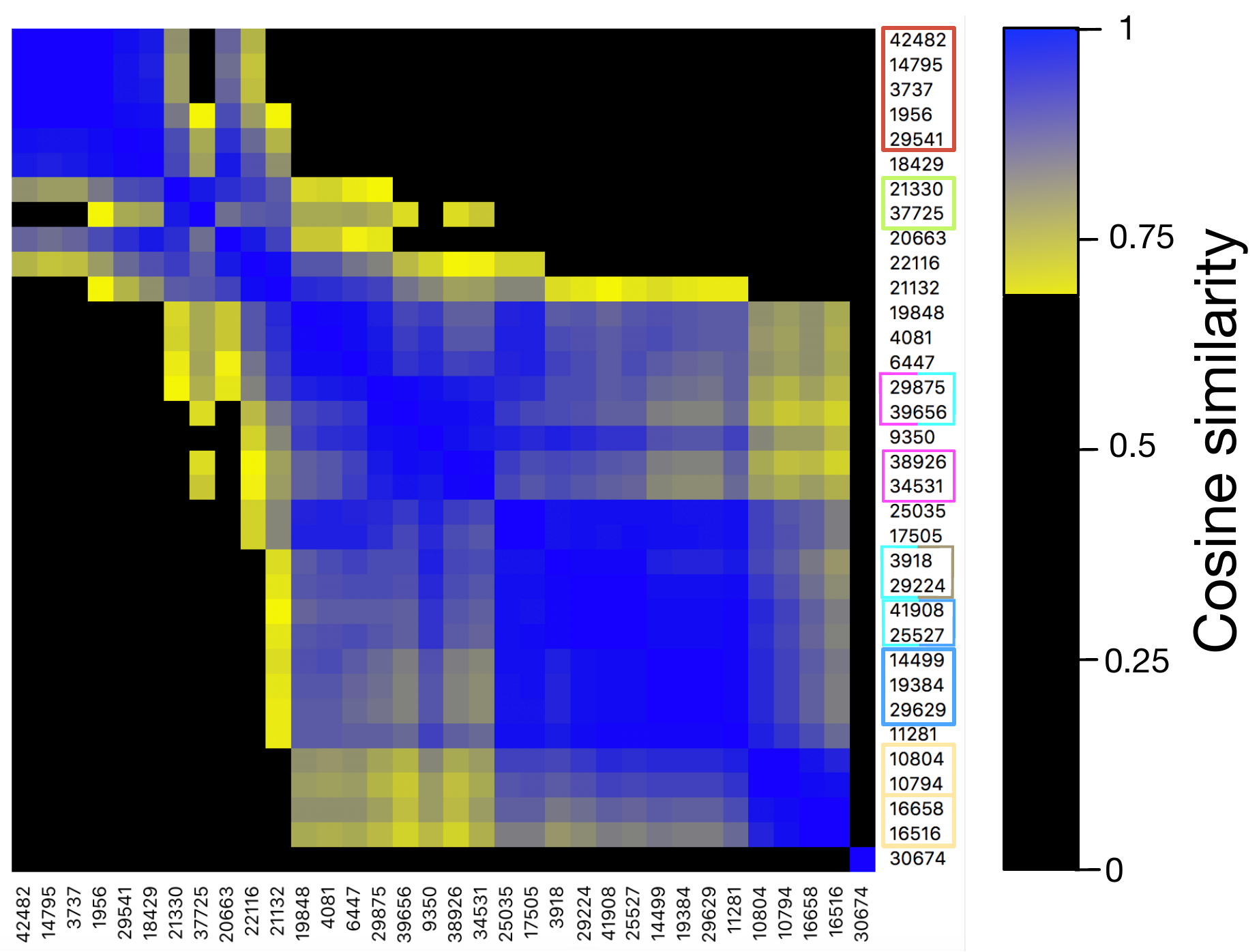}
 \caption{Clustered cosine similarity matrix (genres)}
 \label{fig:clusters_cosine}
   \end{subfigure}
   
      \caption{Clustered similarity matrices for top and random users of Figure~\ref{fig:jac_mds}. Plots show heatmap output of a hierarchical clustering algorithm using a) Jaccard distance between users profiles (listened Last.fm artists), b) Cosine distance genres-based profile. In both figures the colored boxes highlight the clusters present in Figure~\ref{fig:jac_mds}. User ids are automatically arranged by the clustering algorithm so as to maximize the sum of similarities of adjacent elements.}
      \label{fig:clustered_similarity_matrices}
\end{figure} 

To complement Figure~\ref{fig:jac_mds}, in Figure~\ref{fig:clustered_similarity_matrices}, we show heatmaps of two similarity matrices, Jaccard similarities of artist-based profiles  (Fig.~\ref{fig:clusters_jaccard}) and Cosine similarities of genre-based profiles (which span the range between 0 and 1, being vectors with only positive components), in Figure~\ref{fig:clusters_cosine}, arranged using the output of a hierarchical clustering algorithm with \emph{ordered leaves} representation~\cite{bar2001fast}, \ie maximizing the sum of similarities among adjacent elements. We see that the clustering structure is apparent, and preserved, in both matrices (see the colored boxes on the user ids), although way clearer in the Jaccard matrix. This is expected as artists-based profiles have a far more granular resolution, and are therefore sparser with respect to genre-based profiles, especially considering that the main genres are almost the same for all users. This is also a witness of the fact that artist information is more suitable to achieve real personalized recommendations than genres, which motivates our choice of artists as user profile features.

\section{Conclusion}
In this paper, we presented scalable approximation algorithms for Jaccard-based similarity search in dynamic data streams.
Specifically, we showed how to sketch the Jaccard similarity via a black box reduction to $\ell_0$ norm estimation, and we gave a locality sensitive hashing scheme that quickly filters out low-similarity pairs.
To the best of our knowledge, these are the first algorithms that can handle item deletions.
In addition to theoretical guarantees, we showed that the algorithm has competitive running times to the established min-hashing approaches.
We also have reason to believe that the algorithm can be successfully applied in real-world applications, as evidenced by its performance for finding Last.fm users with similar musical tastes. 

It would be interesting to extend these ideas for other similarity measures. Though we focused mainly on the Jaccard-index, our approach works for any set-based similarity measure supporting an LSH, compare the characterization of Chierichetti and Kumar~\cite{ChK15}.
It is unclear whether our techniques may be reused for other similarities applied in collaborative filtering.

\bibliographystyle{unsrt}
\bibliography{bibliography}
\end{document}